\documentclass[11pt,a4paper]{article}
\usepackage[T1]{fontenc}
\usepackage[english]{babel}
\usepackage[dvips]{graphicx}
\usepackage{epsf}
\usepackage{amssymb}
\usepackage{pstricks}
\usepackage{pst-node}
\usepackage{lineno}
\usepackage[normalem]{ulem}

\usepackage{tabularx}
\usepackage{booktabs}
\usepackage{multirow}

\newenvironment{block}[3]{%
  \vspace{2mm}\par%
  \refstepcounter{#2}%
  \noindent\textbf{%
    #1 \thesection.\arabic{#2}.
  }%
}{%
  \par%
}

\newcounter{theorem}[section]
\newenvironment{theorem}[1][]{%
  \begin{block}{\textsc{Theorem}}{theorem}{#1}%
  \slshape%
}{%
  \end{block}%
  \vspace{2mm}%
}
\newenvironment{lemma}[1][]{%
  \begin{block}{\textsc{Lemma}}{theorem}{#1}%
  \slshape%
}{%
  \end{block}%
  \vspace{2mm}%
}

\newenvironment{corollary}[1][]{%
  \begin{block}{\textsc{Corollary}}{theorem}{#1}%
  \slshape%
}{%
  \end{block}%
  \vspace{2mm}%
}

\newcounter{invariant}

\newenvironment{proof}[1][Proof]{%
  \par\noindent{\sl{#1}}.
}{%
  \eop
  \bigskip%
}

\newcommand{\eop}{\hspace*{\fill}\nolinebreak$\Box$\nolinebreak\par}

\newcommand{\old}[1]{{}}
\newcommand{\todo}[1]{{\noindent$\times\times\times\times\times\times\times\times\times\times\times\times\times\times\times\times\times\times\times\times\times\times\times\times\times\times\times\times\times\times\times\times\times\times\times\times\times\times\times\times\times\times\times\times\times\times\times\times\times\times\times\times$ $\times\times\times\times\times\times\times\times\times\times\times\times\times\times\times\times\times\times\times\times\times\times\times\times\times\times\times\times\times\times\times\times\times\times\times\times\times\times\times\times\times\times\times\times\times\times\times\times\times\times\times\times$ $\times\times\times\times\times\times\times\times\times\times\times\times\times\times\times\times\times\times\times\times\times\times\times\times\times\times\times\times\times\times\times\times\times\times\times\times\times\times\times\times\times\times\times\times\times\times\times\times\times\times\times\times$}}

\newrgbcolor{gray5}{0.5 0.5 0.5}
\newrgbcolor{gray7}{0.7 0.7 0.7}
\newrgbcolor{gray9}{0.9 0.9 0.9}

\newrgbcolor{gray8}{0.85 0.85 0.85}

\def\eps{\varepsilon}
\def\cx{{\rm x}}
\def\cy{{\rm y}}
\def\pq{\langle p,q \rangle}

\renewcommand\caption[1]{\small\refstepcounter{figure}%
\begin{center}\textbf{Fig.\ \thefigure .}\ #1\end{center}\normalsize}

\begin{document}

%\linenumbers
\title{The Rectilinear Steiner Forest Arborescence problem}

~
\vspace{15mm}
\begin{center}
\noindent\textbf{\Large\uppercase{The Rectilinear Steiner \\[4mm] Forest Arborescence problem}}
\end{center}

\vspace{15mm}
\noindent\textsc{Łukasz Mielewczyk}\\[1mm]
Institute of Informatics\\
University of Gda\'{n}sk, 80-308 Gda\'nsk, Poland\\
\small \texttt{lukasz.mielwczyk@ug.edu.pl}
\\[3mm]
\noindent\textsc{Leonidas Palios}\\[1mm]
Department of Computer Science and Engineering\\
University of Ioannina, GR-45110 Ioannina, Greece\\
\small \texttt{palios@cs.uoi.gr}\\[3mm]
\textsc{Pawe\l{}~\.Zyli\'nski}\\[1mm]
Institute of Informatics\\
University of Gda\'{n}sk, 80-308 Gda\'nsk, Poland\\
\small \texttt{pawel.zylinski@ug.edu.pl}
%\\

\begin{abstract}\small
\noindent{}Let $r$ be a point in the first quadrant $Q_1$ of the plane $\mathbb{R}^2$ and let $P \subset Q_1$ be a set of points such that for any $p \in P$, its $x$- and $y$-coordinate is at least as that of $r$. A rectilinear Steiner arborescence for $P$ with the root $r$ is a rectilinear Steiner tree $T$ for $P \cup \{r\}$ such that for each point $p \in P$, the length of the (unique) path in $T$ from $p$ to the root~$r$ equals $({\rm x}(p)-{\rm x}(r))+({\rm y}(p))-({\rm y}(r))$, where ${\rm x}(q)$ and ${\rm y}(q)$ denote the $x$- and $y$-coordinate, respectively, of point~$q \in P \cup \{r\}$.

Given two point sets $P$ and $R$ lying in the first quadrant $Q_1$ and such that \mbox{$(0,0) \in R$}, the Rectilinear Steiner Forest Arborescence (RSFA) problem is to find the minimum-length spanning forest $F$ such that each connected component $F$ is a rectilinear Steiner arborescence rooted at some root in $R$. The RSFA problem is a natural generalization of the Rectilinear Steiner Arborescence problem, where $R=\{(0,0)\}$, and thus it is NP-hard.
In this paper, we provide a simple exact exponential time algorithm
for the RSFA problem, design a polynomial time approximation scheme as well as a fixed-parameter algorithm.
\end{abstract}

\noindent{\small{\bf Keywords}:\\
\phantom{aaa}rectilinear Steiner arborescence\\ \phantom{aaa}dynamic programming\\ \phantom{aaa}exact algorithm\\
\phantom{aaa}PTAS\\ \phantom{aaa}FPT}


\newpage

\section{Introduction}
\normalsize
A {\sl rectilinear} graph $G=(V(G),E(G))$ is a %connected
plane graph with (weighted) edges corresponding to horizontal or vertical line segments that connect two vertices in the plane $\mathbb{R}^2$, intersecting only at their endpoints. The weight~$l(e)$ of an edge $e \in E(G)$ is equal to the length of the segment it corresponds to, the {\sl weight} of $G$, denoted by $l(G)$, is defined as the sum of all edge weights in $G$, and the union of all the line segments (edges) in $G$, regarded as a point set, is denoted by ${\rm Un}(G)$.\footnote{However, sometimes, we shall leave the notation ${\rm Un}(\cdot)$ out, when this does not lead to misunderstanding.} A {\sl rectilinear Steiner tree} for a set $P$ of points in the plane is a rectilinear acyclic graph such that each point in $P$ is an endpoint of some edge in the tree.

Let $r$ be a point in the first quadrant $Q_1$ of the plane $\mathbb{R}^2$ and let $P \subset Q_1$ be a set of points such that for any $p \in P$, its $x$- and $y$-coordinate is at least as that of $r$. A {\sl rectilinear Steiner arborescence} (RSA) for $P$ with the root $r$ is a rectilinear Steiner tree $T$ for $P \cup \{r\}$ such that for each point $p \in P$, the length of the (unique) path in $T$ from $p$ to the root~$r$ equals $({\rm x}(p)-{\rm x}(r))+({\rm y}(p))-({\rm y}(r))$, where ${\rm x}(q)$ and ${\rm y}(q)$ denote the $x$- and $y$-coordinate, respectively, of point~$q \in P \cup \{r\}$. A {\sl minimum rectilinear Steiner arborescence} (MRSA) for $P$ with the root $r$ is an RSA for $P$ with the root $r$ that has the minimum weight $l(\cdot)$ over all RSA's for $P$ (with the root $r$).

The problem of determining the minimum rectilinear Steiner arborescence --- having  applications in the field of performance-driven VLSI design~\cite{CKL-1999,CLZ-1993} --- was first studied by Nastansky et al.~\cite{NSS-1974} who proposed an integer programming formulation with exponential time complexity. In 1979, Laderia de Matos~\cite{LdM-1979} proposed an exponential time dynamic programming algorithm. In 1985, Trubin~\cite{T-1985} claimed that the problem is polynomially solvable, however, in 1992, Rao et al.~\cite{RSHS-1992} showed Trubin's algorithm to be incorrect; Rao et al.\ also provided a simple greedy 2-approximation for this problem, with $O(n \log n)$ running time, and their result was extended by Cordova and Lee~\cite{CL-1994} to the more general case where points can be located in all four quadrants, with the root located at the origin of $\mathbb{R}^2$. Another fast $2$-approximation, even in the presence of obstacles, was presented by Ramanth~\cite{R-2003}. Some other optimal exponential time algorithms can be found in~\cite{HKMS-1992, LC-1997}, and some heuristics in~\cite{AR-1996,CKL-1999,TS-1995}. More recently,~\cite{LR-2000,Z-2000} presented polynomial time approximation schemes for the problem, and the NP-hardness of the problem was shown by Shi and Su~\cite{SS-2006}. In addition, Fomin et al.~\cite{FKLPS-2020} proposed the first subexponential algorithm for the problem, inspired by the work of Klein and Marx~\cite{KM14}, who obtained a subexponential algorithm for Subset Traveling Salesman Problem on planar graphs.  The depth-restricted variant of the rectilinear Steiner arborescence problem was studied by Ma\ss{}berg in~\cite{M-2015}, and the angle-restricted Steiner arborescence problem was studied in~\cite{BSV-2011,VBS-2011}. Finally, a more general problem, the Generalized Minimum Manhattan Network problem, was defined by Chepoi et al.~\cite{ChNV-2008} and then studied in~\cite{DFKSVW-2013,MOY21}.

\paragraph{The rectilinear Steiner arborescence with the prespecified sub-RSA.}
From a practical point of view, in real life, when considering optimization problems, we sometimes have to deal with scenarios in which we want to temporarily add extra points without however modifying the solution that we already had. Therefore, it is natural to consider such a scenario for the RSA problem as well.

Suppose that we are given a (not necessarily minimum) rectilinear Steiner arborescence $T$ for a point set $P$ with the root $r$, so-called {\sl prespecified sub-RSA}, and let $P^\ast$ be a set of points in the bounding box ${\rm BB}(T)$ of $T$; w.l.o.g.\ we assume that none of the points in $P^\ast$ lies on an edge of $T$. Now, we would like to connect the new points in $P^\ast$ to the arborescence~$T$, thus getting a (rectilinear) arborescence for all points in $P \cup P^\ast$ (see Fig.~\ref{fig:example-preRSA}). Formally, we are interested in computing the minimum rectilinear Steiner arborescence $T^\ast$ for $P \cup P^\ast$ with the root $r$ such that ${\rm Un}(T) \subset {\rm Un}(T^\ast)$.

\begin{figure}
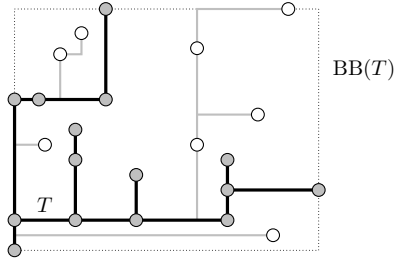

\begin{center}
\pspicture(4,3.5)
%\psgrid

\scalebox{0.8}{
\psline[linewidth=0.5pt,linestyle=dotted,dotsep=1pt](0,0)(5,0)(5,4)(0,4)(0,0)

\rput(5.75,3){\small{${\rm BB}(T)$}}
\rput(.5,.75){\small{$T$}}

\psline[linewidth=1pt,linecolor=lightgray](0.75,2.5)(0.75,3.25)(1.1,3.25)(1.1,3.6)
\psline[linewidth=1pt,linecolor=lightgray](4.25,.25)(0,.25)
\psline[linewidth=1pt,linecolor=lightgray](0,1.75)(0.5,1.75)
\psline[linewidth=1pt,linecolor=lightgray](3,.5)(3,4)(4.5,4)
\psline[linewidth=1pt,linecolor=lightgray](3,2.25)(4,2.25)

\psline[linewidth=1.5pt](0,0)(0,2.5)(1.5,2.5)(1.5,4)
\psline[linewidth=1.5pt](0,0.5)(1,0.5)(1,2)
\psline[linewidth=1.5pt](1,0.5)(3.5,0.5)(3.5,1.5)
\psline[linewidth=1.5pt](3.5,1)(5,1)
\psline[linewidth=1.5pt](2,0.5)(2,1.25)

\cnode[linewidth=0.5pt,fillstyle=solid,fillcolor=lightgray](0,0){3pt}{v}
\cnode[linewidth=0.5pt,fillstyle=solid,fillcolor=lightgray](0,0.5){3pt}{v}
\cnode[linewidth=0.5pt,fillstyle=solid,fillcolor=lightgray](0,2.5){3pt}{v}
\cnode[linewidth=0.5pt,fillstyle=solid,fillcolor=lightgray](0.4,2.5){3pt}{v}
\cnode[linewidth=0.5pt,fillstyle=solid,fillcolor=lightgray](1,0.5){3pt}{v}
\cnode[linewidth=0.5pt,fillstyle=solid,fillcolor=lightgray](1,1.5){3pt}{v}
\cnode[linewidth=0.5pt,fillstyle=solid,fillcolor=lightgray](1,2){3pt}{v}
\cnode[linewidth=0.5pt,fillstyle=solid,fillcolor=lightgray](1.5,2.5){3pt}{v}
\cnode[linewidth=0.5pt,fillstyle=solid,fillcolor=lightgray](1.5,4){3pt}{v}
\cnode[linewidth=0.5pt,fillstyle=solid,fillcolor=lightgray](3.5,0.5){3pt}{v}
\cnode[linewidth=0.5pt,fillstyle=solid,fillcolor=lightgray](3.5,1){3pt}{v}
\cnode[linewidth=0.5pt,fillstyle=solid,fillcolor=lightgray](3.5,1.5){3pt}{v}
\cnode[linewidth=0.5pt,fillstyle=solid,fillcolor=lightgray](5,1){3pt}{v}
\cnode[linewidth=0.5pt,fillstyle=solid,fillcolor=lightgray](2,1.25){3pt}{v}
\cnode[linewidth=0.5pt,fillstyle=solid,fillcolor=lightgray](2,.5){3pt}{v}

\cnode[linewidth=0.5pt,fillstyle=solid,fillcolor=white](0.75,3.25){3pt}{v}
\cnode[linewidth=0.5pt,fillstyle=solid,fillcolor=white](0.5,1.75){3pt}{v}
\cnode[linewidth=0.5pt,fillstyle=solid,fillcolor=white](1.1,3.6){3pt}{v}
\cnode[linewidth=0.5pt,fillstyle=solid,fillcolor=white](4.25,.25){3pt}{v}
\cnode[linewidth=0.5pt,fillstyle=solid,fillcolor=white](3,1.75){3pt}{v}
\cnode[linewidth=0.5pt,fillstyle=solid,fillcolor=white](4.5,4){3pt}{v}
\cnode[linewidth=0.5pt,fillstyle=solid,fillcolor=white](4,2.25){3pt}{v}
\cnode[linewidth=0.5pt,fillstyle=solid,fillcolor=white](3,3.35){3pt}{v}
}
\endpspicture
\caption{The RSA problem with the prespecified sub-RSA:\\
an optimal solution is marked with gray.}\label{fig:example-preRSA}
\end{center}
\end{figure}

The above problem formulation is natural, however, when we start with examining some basic properties of an (optimal) solution $T^\ast$ (Lemma~1.\ref{lem:preRSA_in_H} and discussion below it), we are driven to a property which motivates us to introduce a more natural and general problem: the {\sl Rectilinear Steiner Forest Arborescence problem} which we shall discuss in the next paragraph.

\paragraph{The rectilinear Steiner forest arborescence.}
Let us start with some definitions. Let $T=(V(T),E(T))$ be a rectilinear Steiner arborescence for a point set $P$ with the root $r=(0,0)$ in the plane and let $P^\ast$ be a set of points in the bounding box ${\rm BB}(T)$ of $T$. For a vertex $v \in V(T)$, let $\deg_T(v)$ denote the {\sl degree} of $v$ in $T$, that is, the number of edges in $E(T)$ that $v$ is incident to. A vertex $v \in V(T)$ is called {\sl essential}\, if either $\deg_T(v)=1$ or $\deg_T(v)=2$ and  its two adjacent edges (segments) are not parallel; notice that since $T$ is an RSA, it may have degree 2 vertices that are not essential. For a point $p \in P^\ast \cup V(T)$, define its {\sl right extension} as the maximal horizontal line segment~$l$ with the left endpoint at $p$ such that $l$ has no points in common, except its endpoints, with ${\rm Un}(T)$ and the boundary of ${\rm BB}(T)$; the {\sl left, upward} and {\sl downward extensions}, respectively, of $p$ are defined analogously, see Fig.~\ref{fig:vertical-extension-Hanan-subgrid}(a) for an illustration. Notice that an extension may degenerate to a single point.

\begin{figure}[!b]
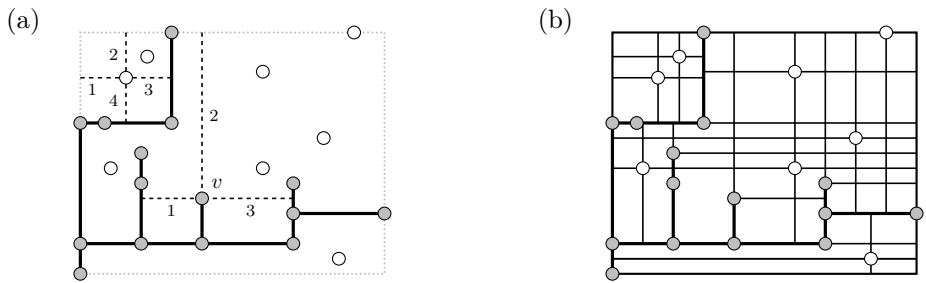

\begin{center}
\pspicture(7,3.75)
%\psgrid

\rput(-0.75,3.35){\small{(a)}}

\scalebox{0.8}{
\psline[linewidth=1pt,linestyle=dotted,dotsep=1pt,linecolor=lightgray](0,0)(5,0)(5,4)(0,4)(0,0)

\psline[linewidth=1.5pt](0,0)(0,2.5)(1.5,2.5)(1.5,4)
\psline[linewidth=1.5pt](0,0.5)(1,0.5)(1,2)
\psline[linewidth=1.5pt](1,0.5)(3.5,0.5)(3.5,1.5)
\psline[linewidth=1.5pt](3.5,1)(5,1)
\psline[linewidth=1.5pt](2,0.5)(2,1.25)

\psline[linewidth=0.75pt,linestyle=dashed, dash=2pt 2pt](0,3.25)(1.5,3.25)
\psline[linewidth=0.75pt,linestyle=dashed, dash=2pt 2pt](0.75,4)(0.75,2.5)

%\psline[linewidth=1pt,linestyle=dashed, dash=2pt 2pt](4.5,4)(4.5,1)
%\psline[linewidth=1pt,linestyle=dashed, dash=2pt 2pt](5,4)(1.5,4)
%\psline[linewidth=1pt,linestyle=dashed, dash=2pt 2pt](0,4)(0,2.5)

\rput(2.25,1.5){\small{$v$}}

\rput(1.5,1.05){\scriptsize{$1$}}
\rput(2.2,2.625){\scriptsize{$2$}}
\rput(0.55,3.625){\scriptsize{$2$}}
\rput(0.2,3.05){\scriptsize{$1$}}
\rput(.55,2.875){\scriptsize{$4$}}
\rput(2.8,1.05){\scriptsize{$3$}}
\rput(1.125,3.05){\scriptsize{$3$}}

%\psline[linewidth=1pt,linestyle=dashed, dash=2pt 2pt](1.5,2.5)(5,2.5)
\psline[linewidth=0.75pt,linestyle=dashed, dash=2pt 2pt](2,4)(2,1.25)(3.5,1.25)
\psline[linewidth=0.75pt,linestyle=dashed, dash=2pt 2pt](2,1.25)(1,1.25)

\cnode[linewidth=0.5pt,fillstyle=solid,fillcolor=lightgray](0,0){3pt}{v}
\cnode[linewidth=0.5pt,fillstyle=solid,fillcolor=lightgray](0,0.5){3pt}{v}
\cnode[linewidth=0.5pt,fillstyle=solid,fillcolor=lightgray](0,2.5){3pt}{v}
\cnode[linewidth=0.5pt,fillstyle=solid,fillcolor=lightgray](0.4,2.5){3pt}{v}
\cnode[linewidth=0.5pt,fillstyle=solid,fillcolor=lightgray](1,0.5){3pt}{v}
\cnode[linewidth=0.5pt,fillstyle=solid,fillcolor=lightgray](1,1.5){3pt}{v}
\cnode[linewidth=0.5pt,fillstyle=solid,fillcolor=lightgray](1,2){3pt}{v}
\cnode[linewidth=0.5pt,fillstyle=solid,fillcolor=lightgray](1.5,2.5){3pt}{v}
\cnode[linewidth=0.5pt,fillstyle=solid,fillcolor=lightgray](1.5,4){3pt}{v}
\cnode[linewidth=0.5pt,fillstyle=solid,fillcolor=lightgray](3.5,0.5){3pt}{v}
\cnode[linewidth=0.5pt,fillstyle=solid,fillcolor=lightgray](3.5,1){3pt}{v}
\cnode[linewidth=0.5pt,fillstyle=solid,fillcolor=lightgray](3.5,1.5){3pt}{v}
\cnode[linewidth=0.5pt,fillstyle=solid,fillcolor=lightgray](5,1){3pt}{v}
\cnode[linewidth=0.5pt,fillstyle=solid,fillcolor=lightgray](2,1.25){3pt}{v}
\cnode[linewidth=0.5pt,fillstyle=solid,fillcolor=lightgray](2,.5){3pt}{v}

\cnode[linewidth=0.5pt,fillstyle=solid,fillcolor=white](0.75,3.25){3pt}{v}
\cnode[linewidth=0.5pt,fillstyle=solid,fillcolor=white](0.5,1.75){3pt}{v}
\cnode[linewidth=0.5pt,fillstyle=solid,fillcolor=white](1.1,3.6){3pt}{v}
\cnode[linewidth=0.5pt,fillstyle=solid,fillcolor=white](4.25,.25){3pt}{v}
\cnode[linewidth=0.5pt,fillstyle=solid,fillcolor=white](3,1.75){3pt}{v}
\cnode[linewidth=0.5pt,fillstyle=solid,fillcolor=white](4.5,4){3pt}{v}
\cnode[linewidth=0.5pt,fillstyle=solid,fillcolor=white](4,2.25){3pt}{v}
\cnode[linewidth=0.5pt,fillstyle=solid,fillcolor=white](3,3.35){3pt}{v}
}
\endpspicture
\pspicture(4,3.2)
%\psgrid
\rput(-0.75,3.35){\small{(b)}}

\scalebox{0.8}{
\psline[linewidth=1pt](0,0)(5,0)(5,4)(0,4)(0,0)

\psline[linewidth=1.5pt](0,0)(0,2.5)(1.5,2.5)(1.5,4)
\psline[linewidth=1.5pt](0,0.5)(1,0.5)(1,2)
\psline[linewidth=1.5pt](1,0.5)(3.5,0.5)(3.5,1.5)
\psline[linewidth=1.5pt](3.5,1)(5,1)
\psline[linewidth=1.5pt](2,0.5)(2,1.25)

\psline[linewidth=0.75pt](0,3.25)(1.5,3.25)
\psline[linewidth=0.75pt](0.75,4)(0.75,2.5)
\psline[linewidth=0.75pt](0,1.75)(1,1.75)
\psline[linewidth=0.75pt](0.5,.5)(0.5,2.5)

\psline[linewidth=0.75pt](1.5,2.5)(5,2.5)
\psline[linewidth=0.75pt](2,4)(2,1.25)(3.5,1.25)

\psline[linewidth=0.75pt](4.25,0)(4.25,1)
\psline[linewidth=0.75pt](0,0.25)(5,0.25)

\psline[linewidth=0.75pt](3.5,0.5)(5,0.5)

\psline[linewidth=0.75pt](0,0.25)(5,0.25)

\psline[linewidth=0.75pt](1,2.5)(1,2)(5,2)
\psline[linewidth=0.75pt](3.5,4)(3.5,1.5)(5,1.5)

%\psline[linewidth=1pt](3,4)(3,0.5)
\psline[linewidth=0.75pt](1,1.75)(5,1.75)

%\psline[linewidth=1pt](1.5,3.75)(5,3.75)
\psline[linewidth=0.75pt](4.5,4)(4.5,1)

\psline[linewidth=0.75pt](0,2.25)(5,2.25)
\psline[linewidth=0.75pt](4,1)(4,4)

\psline[linewidth=0.75pt](3,4)(3,0.5)
\psline[linewidth=0.75pt](1.5,3.35)(5,3.35)

\psline[linewidth=0.75pt](1.1,4)(1.1,2.5)
\psline[linewidth=0.75pt](0,3.6)(1.5,3.6)

\cnode[linewidth=0.5pt,fillstyle=solid,fillcolor=lightgray](0,0){3pt}{v}
\cnode[linewidth=0.5pt,fillstyle=solid,fillcolor=lightgray](0,0.5){3pt}{v}
\cnode[linewidth=0.5pt,fillstyle=solid,fillcolor=lightgray](0,2.5){3pt}{v}
\cnode[linewidth=0.5pt,fillstyle=solid,fillcolor=lightgray](0.4,2.5){3pt}{v}
\cnode[linewidth=0.5pt,fillstyle=solid,fillcolor=lightgray](1,0.5){3pt}{v}
\cnode[linewidth=0.5pt,fillstyle=solid,fillcolor=lightgray](1,1.5){3pt}{v}
\cnode[linewidth=0.5pt,fillstyle=solid,fillcolor=lightgray](1,2){3pt}{v}
\cnode[linewidth=0.5pt,fillstyle=solid,fillcolor=lightgray](1.5,2.5){3pt}{v}
\cnode[linewidth=0.5pt,fillstyle=solid,fillcolor=lightgray](1.5,4){3pt}{v}
\cnode[linewidth=0.5pt,fillstyle=solid,fillcolor=lightgray](3.5,0.5){3pt}{v}
\cnode[linewidth=0.5pt,fillstyle=solid,fillcolor=lightgray](3.5,1){3pt}{v}
\cnode[linewidth=0.5pt,fillstyle=solid,fillcolor=lightgray](3.5,1.5){3pt}{v}
\cnode[linewidth=0.5pt,fillstyle=solid,fillcolor=lightgray](5,1){3pt}{v}
\cnode[linewidth=0.5pt,fillstyle=solid,fillcolor=lightgray](2,1.25){3pt}{v}
\cnode[linewidth=0.5pt,fillstyle=solid,fillcolor=lightgray](2,.5){3pt}{v}

\cnode[linewidth=0.5pt,fillstyle=solid,fillcolor=white](0.75,3.25){3pt}{v}
\cnode[linewidth=0.5pt,fillstyle=solid,fillcolor=white](0.5,1.75){3pt}{v}
\cnode[linewidth=0.5pt,fillstyle=solid,fillcolor=white](1.1,3.6){3pt}{v}
\cnode[linewidth=0.5pt,fillstyle=solid,fillcolor=white](4.25,.25){3pt}{v}
\cnode[linewidth=0.5pt,fillstyle=solid,fillcolor=white](3,1.75){3pt}{v}
\cnode[linewidth=0.5pt,fillstyle=solid,fillcolor=white](4.5,4){3pt}{v}
\cnode[linewidth=0.5pt,fillstyle=solid,fillcolor=white](4,2.25){3pt}{v}
\cnode[linewidth=0.5pt,fillstyle=solid,fillcolor=white](3,3.35){3pt}{v}
}
\endpspicture
\caption{(a) An illustration
of the left (1), upward (2), right (3) and downward (4) extensions. Notice that the downward extension of $v$
degenerates to a single point.  (b) The subgrid ${\rm H}(T,P^\ast)$.}\label{fig:vertical-extension-Hanan-subgrid}
\end{center}
\end{figure}

Let ${\rm H}(T,P^\ast)$ be the subgrid consisting of all points of (i) the boundary of the bounding box ${\rm BB}(T)$, (ii) the edges of $T$, (iii) all extensions of all points in $P^\ast$, and (iv) only right and upward extensions of all essential vertices of $T$; see Fig.~\ref{fig:vertical-extension-Hanan-subgrid}(b) for an illustration. We have the following lemma.
\begin{lemma}\label{lem:preRSA_in_H}
There exists a minimum rectilinear Steiner arborescence
for $P \cup P^\ast$ with the~root~$r$ and with the prespecified sub-RSA $T$ for $P$ with the~root~$r$
such that it uses only subsegments of the grid ${\rm H}(T,P^\ast)$.
\end{lemma}

\begin{proof}
It follows by arguments similar to those in the proof of Lemma~3 in~\cite{CCGJLZ11}. Namely, let $F$ be a minimum rectilinear Steiner arborescence for $P \cup P^\ast$ with the~root~$r$ and with the pre-specified sub-RSA $T$ for $P$ with the~root~$r$. Suppose now that $F$ does not use only subsegments of the grid ${\rm H}(T,P^\ast)$. The idea is to move the solution into ${\rm H}(T,P^\ast)$ by repeating the following steps. While there exists a vertical edge of $F$ not lying on ${\rm H}(T,P^\ast)$, pick a vertical path $\pi \subseteq {\rm Un}(F)$ of maximum length that does not lie on the grid. Observe that by minimality of $F$, it follows that $\pi$ does not cross an edge of $T$. Since the number of horizontal edges touching $\pi$ from the left is at most one (by minimality of $F$) and the number of horizontal edges touching $\pi$ from the right is at least one, move $\pi$ to the right (by creating and/or extending a left edge and shortening all right edges) until $\pi$ is no longer a vertical path of maximal length (and so overlaps with another vertical path) or $\pi$ reaches the grid ${\rm H}(T,P^\ast)$. Clearly, this will not increase the length of $F$. Finally, eliminate horizontal paths not lying on the grid ${\rm H}(T,P^\ast)$ in the same way.
\end{proof}
\begin{figure}
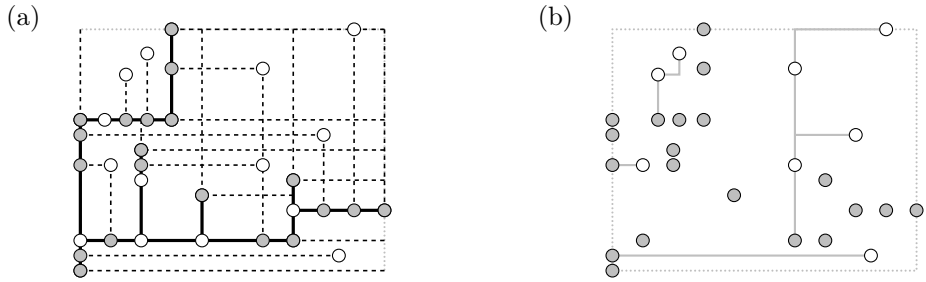

\begin{center}
\pspicture(7,3.2)
%\psgrid
\rput(-0.75,3.35){\small{(a)}}

\scalebox{0.8}{

\psline[linewidth=1pt,linestyle=dotted,dotsep=1pt,linecolor=lightgray](5,0)(5,4)(0,4)
\psline[linewidth=0.75pt,linestyle=dashed, dash=2pt 2pt](0,0)(5,0)
\psline[linewidth=0.75pt,linestyle=dashed, dash=2pt 2pt](0,4)(0,0)

\psline[linewidth=1.5pt](0,0)(0,2.5)(1.5,2.5)(1.5,4)
\psline[linewidth=1.5pt](0,0.5)(1,0.5)(1,2)
\psline[linewidth=1.5pt](1,0.5)(3.5,0.5)(3.5,1.5)
\psline[linewidth=1.5pt](3.5,1)(5,1)
\psline[linewidth=1.5pt](2,0.5)(2,1.25)

%\psline[linewidth=0.75pt,linestyle=dashed, dash=2pt 2pt](0,3.25)(.75,3.25)
\psline[linewidth=0.75pt,linestyle=dashed, dash=2pt 2pt](0.75,2.5)(0.75,3.25)
\psline[linewidth=0.75pt,linestyle=dashed, dash=2pt 2pt](0,1.75)(0.5,1.75)
\psline[linewidth=0.75pt,linestyle=dashed, dash=2pt 2pt](0.5,.5)(0.5,1.75)

\psline[linewidth=0.75pt,linestyle=dashed, dash=2pt 2pt](1.5,2.5)(5,2.5)
\psline[linewidth=0.75pt,linestyle=dashed, dash=2pt 2pt](2,4)(2,1.25)(3.5,1.25)

\psline[linewidth=0.75pt,linestyle=dashed, dash=2pt 2pt](0,0.25)(4.25,0.25)

\psline[linewidth=0.75pt,linestyle=dashed, dash=2pt 2pt](3.5,0.5)(5,0.5)

\psline[linewidth=0.75pt,linestyle=dashed, dash=2pt 2pt](5,1)(5,4)

\psline[linewidth=0.75pt,linestyle=dashed, dash=2pt 2pt](1,2.5)(1,2)(5,2)
\psline[linewidth=0.75pt,linestyle=dashed, dash=2pt 2pt](3.5,4)(3.5,1.5)(5,1.5)

%\psline[linewidth=0.75pt,linestyle=dashed, dash=2pt 2pt](2.5,1.75)(2.5,0.5)
\psline[linewidth=0.75pt,linestyle=dashed, dash=2pt 2pt](1,1.75)(3,1.75)

\psline[linewidth=0.75pt,linestyle=dashed, dash=2pt 2pt](1.5,4)(5,4)
\psline[linewidth=0.75pt,linestyle=dashed, dash=2pt 2pt](4.5,4)(4.5,1)

\psline[linewidth=0.75pt,linestyle=dashed, dash=2pt 2pt](0,2.25)(4,2.25)
\psline[linewidth=0.75pt,linestyle=dashed, dash=2pt 2pt](4,1)(4,2.25)

\psline[linewidth=0.75pt,linestyle=dashed, dash=2pt 2pt](3,3.35)(3,0.5)
\psline[linewidth=0.75pt,linestyle=dashed, dash=2pt 2pt](1.5,3.35)(3,3.35)

\psline[linewidth=0.75pt,linestyle=dashed, dash=2pt 2pt](1.1,3.6)(1.1,2.5)
%\psline[linewidth=0.75pt,linestyle=dashed, dash=2pt 2pt](0,3.6)(1.1,3.6)

\cnode[linewidth=0.5pt,fillstyle=solid,fillcolor=white](0,0){3pt}{v}
\cnode[linewidth=0.5pt,fillstyle=solid,fillcolor=white](0,0.5){3pt}{v}
\cnode[linewidth=0.5pt,fillstyle=solid,fillcolor=white](0,2.5){3pt}{v}
\cnode[linewidth=0.5pt,fillstyle=solid,fillcolor=white](0.4,2.5){3pt}{v}
\cnode[linewidth=0.5pt,fillstyle=solid,fillcolor=white](1,0.5){3pt}{v}
\cnode[linewidth=0.5pt,fillstyle=solid,fillcolor=white](1,1.5){3pt}{v}
\cnode[linewidth=0.5pt,fillstyle=solid,fillcolor=white](1,2){3pt}{v}
\cnode[linewidth=0.5pt,fillstyle=solid,fillcolor=white](1.5,2.5){3pt}{v}
\cnode[linewidth=0.5pt,fillstyle=solid,fillcolor=white](1.5,4){3pt}{v}
\cnode[linewidth=0.5pt,fillstyle=solid,fillcolor=white](3.5,0.5){3pt}{v}
\cnode[linewidth=0.5pt,fillstyle=solid,fillcolor=white](3.5,1){3pt}{v}
\cnode[linewidth=0.5pt,fillstyle=solid,fillcolor=white](3.5,1.5){3pt}{v}
\cnode[linewidth=0.5pt,fillstyle=solid,fillcolor=white](5,1){3pt}{v}
\cnode[linewidth=0.5pt,fillstyle=solid,fillcolor=white](2,1.25){3pt}{v}
\cnode[linewidth=0.5pt,fillstyle=solid,fillcolor=white](2,.5){3pt}{v}

\cnode[linewidth=0.5pt,fillstyle=solid,fillcolor=white](0.75,3.25){3pt}{v}
\cnode[linewidth=0.5pt,fillstyle=solid,fillcolor=white](0.5,1.75){3pt}{v}
\cnode[linewidth=0.5pt,fillstyle=solid,fillcolor=white](1.1,3.6){3pt}{v}
\cnode[linewidth=0.5pt,fillstyle=solid,fillcolor=white](4.25,.25){3pt}{v}
\cnode[linewidth=0.5pt,fillstyle=solid,fillcolor=white](3,1.75){3pt}{v}
\cnode[linewidth=0.5pt,fillstyle=solid,fillcolor=white](4.5,4){3pt}{v}
\cnode[linewidth=0.5pt,fillstyle=solid,fillcolor=white](4,2.25){3pt}{v}
\cnode[linewidth=0.5pt,fillstyle=solid,fillcolor=white](3,3.35){3pt}{v}

\cnode[linewidth=0.5pt,fillstyle=solid,fillcolor=lightgray](0.75,2.5){3pt}{v}
\cnode[linewidth=0.5pt,fillstyle=solid,fillcolor=lightgray](1.1,2.5){3pt}{v}
\cnode[linewidth=0.5pt,fillstyle=solid,fillcolor=lightgray](0,1.75){3pt}{v}
\cnode[linewidth=0.5pt,fillstyle=solid,fillcolor=lightgray](0.5,.5){3pt}{v}
\cnode[linewidth=0.5pt,fillstyle=solid,fillcolor=lightgray](0,2.25){3pt}{v}
\cnode[linewidth=0.5pt,fillstyle=solid,fillcolor=lightgray](4,1){3pt}{v}
\cnode[linewidth=0.5pt,fillstyle=solid,fillcolor=lightgray](1.5,2.5){3pt}{v}
\cnode[linewidth=0.5pt,fillstyle=solid,fillcolor=lightgray](2,1.25){3pt}{v}
\cnode[linewidth=0.5pt,fillstyle=solid,fillcolor=lightgray](0,.25){3pt}{v}
\cnode[linewidth=0.5pt,fillstyle=solid,fillcolor=lightgray](1,2){3pt}{v}
\cnode[linewidth=0.5pt,fillstyle=solid,fillcolor=lightgray](3.5,1.5){3pt}{v}
\cnode[linewidth=0.5pt,fillstyle=solid,fillcolor=lightgray](3,.5){3pt}{v}
\cnode[linewidth=0.5pt,fillstyle=solid,fillcolor=lightgray](1.5,3.35){3pt}{v}
%\cnode[linewidth=0.5pt,fillstyle=solid,fillcolor=lightgray](2.5,.5){3pt}{v}
\cnode[linewidth=0.5pt,fillstyle=solid,fillcolor=lightgray](1,1.75){3pt}{v}
\cnode[linewidth=0.5pt,fillstyle=solid,fillcolor=lightgray](1.5,4){3pt}{v}
\cnode[linewidth=0.5pt,fillstyle=solid,fillcolor=lightgray](4.5,1){3pt}{v}

\cnode[linewidth=0.5pt,fillstyle=solid,fillcolor=lightgray](0,0){3pt}{v}
\cnode[linewidth=0.5pt,fillstyle=solid,fillcolor=lightgray](5,1){3pt}{v}
\cnode[linewidth=0.5pt,fillstyle=solid,fillcolor=lightgray](0,2.5){3pt}{v}

\cnode[linewidth=0.5pt,fillstyle=solid,fillcolor=lightgray](3.5,.5){3pt}{v}

}
\endpspicture
\pspicture(4,3.2)
%\psgrid
\rput(-0.75,3.35){\small{(b)}}

\scalebox{0.8}{

\psline[linewidth=1pt,linestyle=dotted,dotsep=1pt,linecolor=lightgray](0,0)(5,0)(5,4)(0,4)(0,0)

\psline[linewidth=1pt,linecolor=lightgray](0.75,2.5)(0.75,3.25)(1.1,3.25)(1.1,3.6)
\psline[linewidth=1pt,linecolor=lightgray](4.25,.25)(0,.25)
\psline[linewidth=1pt,linecolor=lightgray](0,1.75)(0.5,1.75)
\psline[linewidth=1pt,linecolor=lightgray](3,.5)(3,4)(4.5,4)
\psline[linewidth=1pt,linecolor=lightgray](3,2.25)(4,2.25)

\cnode[linewidth=0.5pt,fillstyle=solid,fillcolor=white](0.75,3.25){3pt}{v}
\cnode[linewidth=0.5pt,fillstyle=solid,fillcolor=white](0.5,1.75){3pt}{v}
\cnode[linewidth=0.5pt,fillstyle=solid,fillcolor=white](1.1,3.6){3pt}{v}
\cnode[linewidth=0.5pt,fillstyle=solid,fillcolor=white](4.25,.25){3pt}{v}
\cnode[linewidth=0.5pt,fillstyle=solid,fillcolor=white](3,1.75){3pt}{v}
\cnode[linewidth=0.5pt,fillstyle=solid,fillcolor=white](4.5,4){3pt}{v}
\cnode[linewidth=0.5pt,fillstyle=solid,fillcolor=white](4,2.25){3pt}{v}
\cnode[linewidth=0.5pt,fillstyle=solid,fillcolor=white](3,3.35){3pt}{v}

\cnode[linewidth=0.5pt,fillstyle=solid,fillcolor=lightgray](0.75,2.5){3pt}{v}
\cnode[linewidth=0.5pt,fillstyle=solid,fillcolor=lightgray](1.1,2.5){3pt}{v}
\cnode[linewidth=0.5pt,fillstyle=solid,fillcolor=lightgray](0,1.75){3pt}{v}
\cnode[linewidth=0.5pt,fillstyle=solid,fillcolor=lightgray](0.5,.5){3pt}{v}
\cnode[linewidth=0.5pt,fillstyle=solid,fillcolor=lightgray](0,2.25){3pt}{v}
\cnode[linewidth=0.5pt,fillstyle=solid,fillcolor=lightgray](4,1){3pt}{v}
\cnode[linewidth=0.5pt,fillstyle=solid,fillcolor=lightgray](1.5,2.5){3pt}{v}
\cnode[linewidth=0.5pt,fillstyle=solid,fillcolor=lightgray](2,1.25){3pt}{v}
\cnode[linewidth=0.5pt,fillstyle=solid,fillcolor=lightgray](0,.25){3pt}{v}
\cnode[linewidth=0.5pt,fillstyle=solid,fillcolor=lightgray](1,2){3pt}{v}
\cnode[linewidth=0.5pt,fillstyle=solid,fillcolor=lightgray](3.5,1.5){3pt}{v}
\cnode[linewidth=0.5pt,fillstyle=solid,fillcolor=lightgray](3,.5){3pt}{v}
\cnode[linewidth=0.5pt,fillstyle=solid,fillcolor=lightgray](1.5,3.35){3pt}{v}
%\cnode[linewidth=0.5pt,fillstyle=solid,fillcolor=lightgray](2.5,.5){3pt}{v}
\cnode[linewidth=0.5pt,fillstyle=solid,fillcolor=lightgray](1,1.75){3pt}{v}
\cnode[linewidth=0.5pt,fillstyle=solid,fillcolor=lightgray](1.5,4){3pt}{v}
\cnode[linewidth=0.5pt,fillstyle=solid,fillcolor=lightgray](4.5,1){3pt}{v}

\cnode[linewidth=0.5pt,fillstyle=solid,fillcolor=lightgray](0,0){3pt}{v}
\cnode[linewidth=0.5pt,fillstyle=solid,fillcolor=lightgray](5,1){3pt}{v}
\cnode[linewidth=0.5pt,fillstyle=solid,fillcolor=lightgray](0,2.5){3pt}{v}

\cnode[linewidth=0.5pt,fillstyle=solid,fillcolor=lightgray](3.5,.5){3pt}{v}

}
\endpspicture
\caption{(a) Grey points constitute the set $R$ of roots: essential vertices of the pre-specified RSA and all essential endpoints of the left/downward extensions of points in $P^\ast$. (b) The RFSA problem corresponding to the RSA problem with the pre-specified RSA in Fig.~\ref{fig:example-preRSA}.}\label{fig:MRFSA}
\end{center}
\end{figure}
The left endpoint of the left extension (resp.\ the bottom endpoint of the downward extension) of a point $p \in P^\ast$ is called {\sl essential} if it is a point of ${\rm Un}(T)$, that is, it lies on an edge of~$T$. Let $R$ be the set of all essential vertices of $T$ and all essential endpoints of the left/downward extensions of points in $P^\ast$ (see Fig.~\ref{fig:MRFSA}(a)); notice that $|R|=O(|P^\ast|+|V(T)|)$. Taking into account the above lemma, one can observe that there exists an optimal solution~$T^\ast$, that is, a minimum rectilinear Steiner arborescence $T^\ast$ for $P \cup P^\ast$ with the~root~$r$ and with the pre-specified sub-RSA $T$ for $P$ (with the~root~$r$), such that each point in $P^\ast$ is connected to $T$ with an rectilinear path whose one endpoint is in $R$. And so, rather then considering the pre-specified rectilinear Steiner arborescence $T$, in particular, its edge set $E(T)$, we may restrict ourselves only to the set $P^\ast$ of points and, as possible roots, to the set of all points in $R$. This observation is crucial, and hence we introduce the following general problem (see Fig.~\ref{fig:MRFSA}(b)).
%
%\newpage
\begin{description}
\item[The Rectilinear Steiner Forest Arborescence (RSFA) problem.]~\\
Given two disjoint sets $P$ and $R$ of $n$ points and $m$ roots, respectively, lying in the first quadrant of $\mathbb{R}^2$ and such that $(0,0) \in R$, find the minimum rectilinear spanning forest $F$, called {\sl a rectilinear Steiner forest arborescence for $P$ with the root set $R$}, such that each of connected components of $F$ is an RSA rooted at some root in $R$, that is, for each point $p \in P$, there exists a root $r \in R$ and a path $\pi$ in $F$ such that the length of $\pi$ is equal to $({\rm x}(p)-{\rm x}(r))+({\rm y}(p)-{\rm y}(r))$.
\end{description}

As alluded above, our restatement of the problem definition asymptotically preserves the size of the input, and thus, any efficient algorithm for the RSFA problem can be applied to the RSA problem with the pre-specified sub-RSA, the problem defined formerly. Next, since the RSA problem is a special case of the RSFA problem where $R=\{(0,0)\}$, bearing in mind~\cite{SS-2006},
we immediately obtain the NP-hardness of the RSFA problem.
\begin{corollary}\label{cor-NP-MRSFA}
The RSFA problem is NP-hard.
\end{corollary}

It is worth pointing out that a restricted variant of the RSFA problem, called the $(1,k)$-MRDPT problem (the Minimum Cost Rectilinear Distance Preserving Tree problem), where $k$ is the number of non-root input points, has been studied in~\cite{TS-1995}, for the purpose of a new efficient heuristic for the four-quadrant RSA problem. In this variant, we have one ``independent'' layer of roots and several independent points in a second layer. And, based upon dynamic programming, T\'ellez and Sarrafzadeh provided a quadratic time algorithm that optimally solves the $(1,k)$-MRDPT problem; see~\cite{TS-1995} for more details.

\paragraph{Our contribution.} 
In this paper, we rely on Corollary~1.\ref{cor-NP-MRSFA} and the algorithmic results for the RSA problem~\cite{CC-2018,LC-1997,LR-2000} in order to provide an exact exponential time algorithm (Section~\ref{sec:Exact}), to propose a PTAS (Section~\ref{sec:PTAS}) and an FPT algorithm (Section~\ref{sec:FPT}). Finally, some future  improvements as well as  some open problems, related to Manhattan networks, are discussed.

\paragraph{Notation.} Unless otherwise stated, all points lie in the first quadrant of $\mathbb{R}^2$ and the origin $(0,0)$ always belongs to the root set $R$. For a point $p$, $||p||$ denotes the sum of coordinates ${\rm x}(p)+{\rm y}(p)$, while the nearest root in $R$ that covers $p$ (ties are broken arbitrarily) is denoted by $r(p)$; notice that if $p$ is a root itself, then $r(p)=p$. For two points $p$ and $q$, the shortest distance between two points $p$ and $q$ (distances are measured in the $L_1$-metric) is denoted with ${\rm dist}(p,q)$,  the point with coordinates $(\min\{{\rm x}(p),{\rm x}(q)\},\min\{{\rm y}(p),{\rm y}(q)\})$ is denoted with $\langle p,q \rangle$, and a point $p$ is said to {\sl cover} $q$ if and only if ${\rm x}(p) \le {\rm x}(q)$ and
${\rm y}(p) \le {\rm y}(q)$. Finally, we shall use an abbreviation RSFA for a rectilinear Steiner forest arborescence.

\section{Exact algorithm}\label{sec:Exact}
In 1992, Ho et al.\ presented an $O(n^2\, 3^n)$ time algorithm for the problem of determining a minimum rectilinear Steiner arborescence of $n$ points in the plane~\cite{HKMS-1992}. Their algorithm is based upon a dynamic programming approach, it originates a classical algorithm for the minimum Steiner tree problem in~\cite{DW72}, and as we show below, its time complexity can be slightly improved to $O(3^n)$, which then leads to an exponential time algorithm that optimally solves the RSFA problem in $O(mn^2+ 3^n)$ time, where $n$ and $m$ is the number of input points and roots, respectively. 

Let $o=(0,0)$ be the origin, let $P$ be a set of $n$ points in the first quadrant $Q_1$ of the plane (excluding the origin~$o$). For a non-empty subset $X \subseteq P$, let $r(X)=(\min_{p \in P}\cx(p),  \min_{p \in P}\cy(p))$ denote its {\sl local root}. We  observe that the local root $r(X)$ of  $X$ can be determined in $\Theta(|X|)$ time, and hence all the subsets of $P$ together with their local roots can be determined in total $O(n\, 2^n) = O(3^n)$ time, in the preprocessing step. (On the other hand,  notice that although there are $2^n$ subsets of $P$, there are only $O(n^2)$ distinct roots in total.) Next, let $A(X)$ denote the minimum rectilinear Steiner arborescence for $X$ with the root $r(X)$. Now, following~\cite{HKMS-1992}, taking into account the definition of the local root, we observe that the length $l(A(X))$ can be expressed by the following dynamic programming recursion (Bellman equation).
$$l(A(X))=\left\{
\begin{array}{ll}
0 & \textnormal{if } |X|=1;\\
\min_{\, U \varsubsetneq X}\ l(A(U)) + {\rm dist}(r(X),r(U)) & \\
{\white \min_{\, U \varsubsetneq X}} + l(A(X \setminus U))+ {\rm dist}(r(X),r(X \setminus U)) & \textnormal{otherwise.}
\end{array}\right.
$$

We tabulate the values $l(A(X))$ %(as well as the arborescence $A(X)$)
with respect to the increasing size of $X$: handling the case $|X|=1$ requires $O(1)$ time, while for $|X| \ge 2$, evaluation of $l(A(X))$ requires consideration of $O(2^{|X|})$ choices of $U$ in evaluating the recursion. Thus, the overall running time is
$$\sum_{k=1}^n \left[{n \choose k} \cdot 2^k \cdot O(1)\right]=O(3^n).$$
Finally, to determine the optimal RSA $A_o(P)$ for $P$ with the root $o=(0,0)$, all we need is to observe that $l(A_o(P))=l(A(P))+{\rm dist}(o,r(P))$. Therefore, we may conclude that the RSA problem (in the first quadrant $Q_1$ of the plane) can be solved in $O(3^n)$ time.

Next, let $P$ and $R$ be two disjoint sets of $n$ points and $m$ roots, respectively, all lying in the first quadrant of $\mathbb{R}^2$ and such that $(0,0) \in R$. For a subset $X \subseteq P$, let $F(X,R)$ denote an optimal rectilinear Steiner forest arborescence for $X$ with the root set $R$; clearly, we are interested in determining $F(P,R)$. Let $H=H(P \cup R)$ be the Hanan grid of $P\cup R$, that is, the grid formed by the vertical and horizontal straight lines passing through all points in $P \cup R$~\cite{H-1966}; notice that each point in $P \cup R$ corresponds to some intersection grid point of $H$.

\begin{lemma}\label{lem:Hanan}
There exists an optimal RSFA for $X$ with the root set $R$ such that it uses only subsegments of the Hanan grid $H(X \cup R) \subseteq H(P \cup R)$.
\end{lemma}

\begin{proof}
It follows by arguments similar to those in the proof of Lemma~1.\ref{lem:preRSA_in_H}; we omit details.
\end{proof}

Taking into account the above lemma, the idea of
our dynamic programming approach is as follows.
First, for each possible location $q$ of the local root of a subset of $P$,
we determine (and tabulate) the root $r(q) \in R$ and the distance ${\rm dist}(q,r(q))$.
Notice that since the coordinates of such local root are defined by 
the relevant $x$- and $y$-coordinates of at most two points in $P$, 
there are at $O(n^2)$ possible locations for roots, and so this step can be done in $O(mn^2)$ time.
Next, for each non-empty subset $X$ of $P$, we determine (and tabulate)
the minimum rectilinear Steiner arborescence $A(X) \subset H(X)$ for $X$ with the root $r(X)$;
following the above approach for the RSA problem,
this can be done in total $O(3^n)$ time.
We also determine (and tabulate) the minimum rectilinear Steiner arborescence
$F^\ast(X,R) \subset H(X \cup R)$ for $X$ with the root $r(r(X)) \in R$;
this can be done in $O(2^n)$ time (since we have already tabulated all $A(X)$'s).

Now, having tabulated all the above data, for each point $p \in P$,
we set $l(F(\{p\},R)) \doteq {\rm dist}(p,r(p))$ and set $F(\{p\},R)$
as any shortest path realizing $l(F(\{p\},R))$ within $H(\{p\} \cup R)$;
this takes $O(n)$ time (since the local root $r(\{p\})=p$ and we have already computed 
${\rm dist}(r(\{p\}),r(p))$ in the first step).
Assume now that for any at most $k$-element subset $Y$ of $P$,
$1 \le k < n$, we have already computed and stored $F(Y,R)$.
Let $Z$ be a $(k+1)$-element subset of $P$.
Observe that $F(Z,R)$ either consists of one connected
component (and so $F(Z,R)=F^\ast(Z,R)$)
or has at least two connected components
(and so $F(Z,R)=F(Z_1,R) \cup F(Z_2,R)$ for some partition $Z_1 \cup Z_2$ of $Z$).
In particular,
$$l(F(Z,R))=\min \{l(F^\ast(Z,R),) \min_{Y \subsetneq Z}(l(F(Y,R))+l(F(Z\setminus Y,R))\}.$$
Since determining $\min_{Y \subsetneq Z}(l(F(Y,R))+l(F(Z\setminus Y,R))\}$ takes $O(2^{k+1})$ time,
the overall time complexity of our approach for determining $F(P,R)$
is $$O(mn^2) +O(3^n)+O(2^n)+O(n)+\sum_{k=2}^n  \left[ {n \choose k} \cdot  2^k \cdot  O(1) \right]=O(mn^2+ 3^n).$$
Therefore, we may conclude with the following theorem.

\begin{theorem}\label{thm:Exact}
The RSFA problem can be solved in $O(mn^2 + 3^n)$ time.
\end{theorem}

Note that if the number~$n$ of points is $O(\log m)$,
we obtain a polynomial time algorithm for the RSFA problem.

\section{Polynomial time approximation scheme}\label{sec:PTAS}

A standard approach for approximation algorithms for geometric optimization problems builds on the techniques developed for polynomial-time approximation schemes (PTAS) for geometric optimization problems due to Arora~\cite{A-1998} and Mitchell~\cite{M-1999}.
In the former approach, to obtain a $(1 + \epsilon)$-approximation in $n^{O(1/\epsilon)}$ time, the instance is first scaled and perturbated to vertices of a polynomial-size integer grid. Next, the grid is recursively partitioned into dissection squares using a quadtree of logarithmic depth, and the so-called Structure Theorem is proved, which guarantees the existence of an almost-optimal solution that crosses the boundary of each dissection square only a few times and only in a number (depending on $\frac{1}{\epsilon}$) of pre-specified portals. Finally, to find a solution satisfying the Structure Theorem, dynamic programming is employed  over the recursive decomposition, and by introducing the concept of portals, the running time becomes $n^c (\log n)^{O(\frac{1}{\epsilon})}$, where $n$ is the size of an input instance, $\epsilon$ is any given constant $0 < \epsilon < 1$, and $c$ is a constant. 

In particular, Arora's method has been successfully applied to solve the RSA problem~\cite{LR-2000,Z-2000},  but without the portal technique and thus obtaining a $(1 + \epsilon)$-approximation only in $n^{O(1/\epsilon)}$ time.\footnote{As observed in~\cite{ChDKN02}, the portal technique seems to be non-applicable to a class of optimization problems with the ,,monotone path'' requirement.} It is natural then to try to apply this technique for the RSFA problem as well. However, an instance depicted in Fig.~\ref{fig:not-Arrora} shows that rounding in Arora's technique may result to optimal solutions that cannot be transformed to the sub-optimal one for the original input, by a small relative cost increase. Consequently, rather then taking an effort in adapting Arora's technique, we make use of the latter aforementioned method, that is, the idea of $c$-guillotine subdivisions proposed by Mitchell~\cite{M-1999}, which has been also successfully applied to the {\sl symmetric} rectilinear Steiner arborescence problem~\cite{CDL-2001}.\footnote{The difference between the RSA problem and its symmetric variant is that in the latter one, the constraint for a path from a point in $P$ to the root is weaker, namely, such a path must be only $y$-monotone~\cite{PS-1985}; in the RSA problem, such a path must be both $x$- and $y$-monotone.}

\begin{figure}
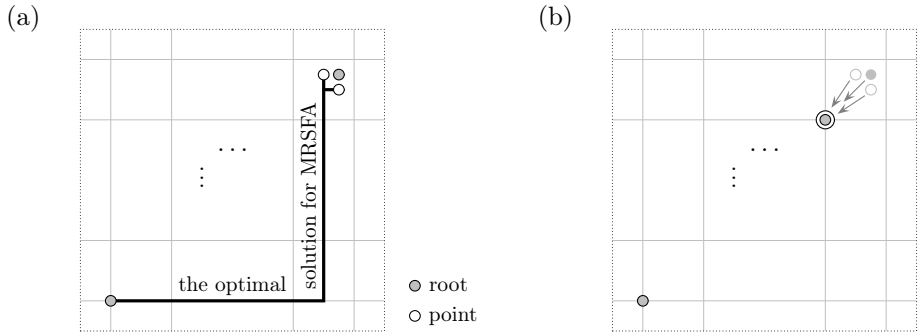

\begin{center}
\pspicture(7,4)
\rput(-0.75,4.15){\small{(a)}}
\scalebox{0.8}{

\psline[linewidth=0.25pt,linecolor=lightgray](0,0.5)(5,0.5)
\psline[linewidth=0.25pt,linecolor=lightgray](0,1.5)(5,1.5)
%\psline[linewidth=0.25pt,linecolor=lightgray](0,2.5)(5,2.5)
\psline[linewidth=0.25pt,linecolor=lightgray](0,3.5)(5,3.5)
\psline[linewidth=0.25pt,linecolor=lightgray](0,4.5)(5,4.5)

\psline[linewidth=0.25pt,linecolor=lightgray](0.5,0)(0.5,5)
\psline[linewidth=0.25pt,linecolor=lightgray](1.5,0)(1.5,5)
%\psline[linewidth=0.25pt,linecolor=lightgray](2.5,0)(2.5,5)
\psline[linewidth=0.25pt,linecolor=lightgray](3.5,0)(3.5,5)
\psline[linewidth=0.25pt,linecolor=lightgray](4.5,0)(4.5,5)

\psline[linewidth=0.5pt,linestyle=dotted,dotsep=1pt](0,0)(5,0)(5,5)(0,5)(0,0)

\psline[linewidth=1.5pt,linecolor=black](4,4.25)(4,0.5)(0.5,0.5)
\psline[linewidth=1.5pt,linecolor=black](4,4)(4.25,4)

\cnode[linewidth=0.5pt,fillstyle=solid,fillcolor=lightgray](0.5,0.5){2.5pt}{v}
\cnode[linewidth=0.5pt,fillstyle=solid,fillcolor=white](4,4.25){2.5pt}{v}
\cnode[linewidth=0.5pt,fillstyle=solid,fillcolor=white](4.25,4){2.5pt}{v}

\cnode[linewidth=0.5pt,fillstyle=solid,fillcolor=lightgray](4.25,4.25){2.5pt}{v}

\cnode[linewidth=0.5pt,fillstyle=solid,fillcolor=lightgray](5.5,0.75){2.5pt}{v}
\cnode[linewidth=0.5pt,fillstyle=solid,fillcolor=white](5.5,0.25){2.5pt}{v}
\rput(6.13,0.25){\small{point}}
\rput(6.05,0.78){\small{root}}

\rput(2.5,.75){\small{the optimal}}
\rput(3.75,2.25){\rotateleft{\small{solution for MRSFA}}}

%\rput(2.5,2){\large$\cdots$}
%\rput(3,2.65){\large$\vdots$}
\rput(2.5,3){\large$\cdots$}
\rput(2,2.65){\large$\vdots$}

}
\endpspicture
\pspicture(4,4)
\rput(-0.75,4.15){\small{(b)}}
\scalebox{0.8}{

\psline[linewidth=0.25pt,linecolor=lightgray](0,0.5)(5,0.5)
\psline[linewidth=0.25pt,linecolor=lightgray](0,1.5)(5,1.5)
%\psline[linewidth=0.25pt,linecolor=lightgray](0,2.5)(5,2.5)
\psline[linewidth=0.25pt,linecolor=lightgray](0,3.5)(5,3.5)
\psline[linewidth=0.25pt,linecolor=lightgray](0,4.5)(5,4.5)

\psline[linewidth=0.25pt,linecolor=lightgray](0.5,0)(0.5,5)
\psline[linewidth=0.25pt,linecolor=lightgray](1.5,0)(1.5,5)
%\psline[linewidth=0.25pt,linecolor=lightgray](2.5,0)(2.5,5)
\psline[linewidth=0.25pt,linecolor=lightgray](3.5,0)(3.5,5)
\psline[linewidth=0.25pt,linecolor=lightgray](4.5,0)(4.5,5)

\psline[linewidth=0.5pt,linestyle=dotted,dotsep=1pt](0,0)(5,0)(5,5)(0,5)(0,0)

%\psline[linewidth=1.5pt,linecolor=black](4,4.25)(4,0.5)(0.5,0.5)
%\psline[linewidth=1.5pt,linecolor=black](4,4)(4.25,4)

\cnode[linewidth=0.5pt,fillstyle=solid,fillcolor=lightgray](0.5,0.5){2.5pt}{v}
\cnode[linewidth=0.5pt,fillstyle=solid,fillcolor=white](3.5,3.5){4.25pt}{v}

\cnode[linewidth=0.5pt,fillstyle=solid,linecolor=lightgray,fillcolor=white](4,4.25){2.5pt}{v}
\cnode[linewidth=0.5pt,fillstyle=solid,linecolor=lightgray,fillcolor=white](4.25,4){2.5pt}{v}
\cnode*[linewidth=0.5pt,fillstyle=solid,linecolor=lightgray,fillcolor=white](4.25,4.25){2.5pt}{v}

\cnode[linewidth=0.5pt,fillstyle=solid,fillcolor=lightgray](3.5,3.5){2.5pt}{v}

%\rput(2.5,2){\large$\cdots$}
%\rput(3,2.65){\large$\vdots$}
\rput(2.5,3){\large$\cdots$}
\rput(2,2.65){\large$\vdots$}

\psline[linewidth=0.5pt,linecolor=black,arrowsize=3pt 2,linecolor=gray]{->}(3.9,4.15)(3.6,3.7)
\psline[linewidth=0.5pt,linecolor=black,arrowsize=3pt 2,linecolor=gray]{->}(4.15,3.9)(3.7,3.6)
\psline[linewidth=0.5pt,linecolor=black,arrowsize=3pt 2,linecolor=gray]{->}(4.15,4.15)(3.8,3.8)

}
\endpspicture

\caption{Rounding may result in the perturbated instance whose optimal solution for the RSFA problem cannot be transformed to the sub-optimal one for the original input, by a small relative cost increase. (b) In the perturbated instance of (a), the non-root (perturbated) points coincides with the (perturbated) root, so the cost of optimal solution for the perturbated instance is $0$.}\label{fig:not-Arrora}
\end{center}
\end{figure}

We follow most of the notation of~\cite{M-1997,M-1998,M-1999}. Let $G=(V(G),E(G))$ be a rectilinear graph, inducing a rectilinear polygonal subdivision of the plane. Without loss of generality, assume that $G$ is restricted to the unit square $B$, i.e.\ ${\rm Un}(G) \subseteq {\rm int}(B)$. (Then, each facet ($2$-face) of $G$ is a bonded rectilinear polygon, possibly with holes.) Let $W$ be a closed axis-aligned rectangle such that $W  \subseteq B$ (we shall refer to $W$ as a {\sl window}), and let $l$ be a axis-aligned line such that $l \cap {\rm int}(W) \neq \emptyset$ (we shall refer to $l$ as a {\sl cut for $G$ with respect to $W$}). The intersection $l \cap {\rm int}(W) \cap {\rm Un}(G)$ consists of a (finite) number of subsegments of $l$. (Some of these ``segments'' may degenerate to a single point when $l$ crosses an edge of $G$.) The endpoints of these subsegments are called the {\sl endpoints along~$l$} (with respect to $W$); we emphasize that the two points where $l$ crosses the boundary of $W$ are not considered to be endpoints along $l$.

\begin{figure}[htb]
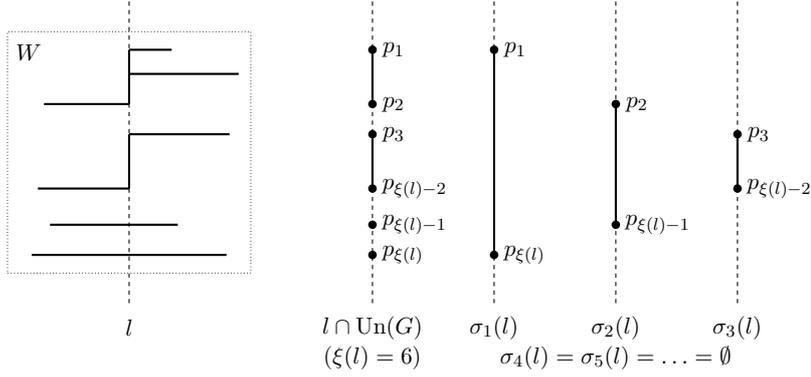

\begin{center}
\pspicture(3.2,3.5)
\scalebox{0.8}{

\psline[linewidth=0.5pt,linestyle=dotted,dotsep=1pt](0,0)(4,0)(4,4)(0,4)(0,0)

\psline[linewidth=0.5pt,linecolor=black,linestyle=dashed, dash=2pt 2pt](2,-0.5)(2,4.5)

\psline[linewidth=1pt,linecolor=black](0.4,0.3)(3.6,0.3)
\psline[linewidth=1pt,linecolor=black](0.7,0.8)(2.8,0.8)
\psline[linewidth=1pt,linecolor=black](0.5,1.4)(2,1.4)
\psline[linewidth=1pt,linecolor=black](2,2.3)(3.65,2.3)
\psline[linewidth=1pt,linecolor=black](0.6,2.8)(2,2.8)
\psline[linewidth=1pt,linecolor=black](2,3.3)(3.8,3.3)
\psline[linewidth=1pt,linecolor=black](2.7,3.7)(2,3.7)

\psline[linewidth=1pt,linecolor=black](2,2.3)(2,1.4)
\psline[linewidth=1pt,linecolor=black](2,3.7)(2,2.8)

\rput(0.35,3.65){$W$}
\rput(2,-0.9){$l$}

}
\endpspicture
\pspicture(1.6,4)
\scalebox{0.8}{

\psline[linewidth=0.5pt,linecolor=black,linestyle=dashed, dash=2pt 2pt](2,-0.5)(2,4.5)

\psline[linewidth=1pt,linecolor=black](2,2.3)(2,1.4)
\psline[linewidth=1pt,linecolor=black](2,3.7)(2,2.8)

\cnode*[linewidth=0.5pt,fillstyle=solid,fillcolor=white](2,0.3){2pt}{v}
\cnode*[linewidth=0.5pt,fillstyle=solid,fillcolor=white](2,.8){2pt}{v}
\cnode*[linewidth=0.5pt,fillstyle=solid,fillcolor=white](2,1.4){2pt}{v}
\cnode*[linewidth=0.5pt,fillstyle=solid,fillcolor=white](2,2.3){2pt}{v}
\cnode*[linewidth=0.5pt,fillstyle=solid,fillcolor=white](2,3.7){2pt}{v}
\cnode*[linewidth=0.5pt,fillstyle=solid,fillcolor=white](2,2.8){2pt}{v}

\rput(2.35,3.7){$p_1$}
\rput(2.35,2.8){$p_2$}
\rput(2.35,2.3){$p_3$}
\rput(2.69,1.4){$p_{\xi(l)-2}$}
\rput(2.69,0.8){$p_{\xi(l)-1}$}
\rput(2.5,0.3){$p_{\xi(l)}$}

\rput(2,-0.9){$l \cap {\rm Un}(G)$}
\rput(2,-1.4){($\xi(l)=6$)}

}
\endpspicture
\pspicture(1.6,4)
\scalebox{0.8}{

\psline[linewidth=0.5pt,linecolor=black,linestyle=dashed, dash=2pt 2pt](2,-0.5)(2,4.5)

%\psline[linewidth=1pt,linecolor=black](2,2.3)(2,1.4)
\psline[linewidth=1pt,linecolor=black](2,3.7)(2,0.3)

\cnode*[linewidth=0.5pt,fillstyle=solid,fillcolor=white](2,0.3){2pt}{v}
%\cnode*[linewidth=0.5pt,fillstyle=solid,fillcolor=white](2,.8){2pt}{v}
%\cnode*[linewidth=0.5pt,fillstyle=solid,fillcolor=white](2,1.4){2pt}{v}
%\cnode*[linewidth=0.5pt,fillstyle=solid,fillcolor=white](2,2.3){2pt}{v}
\cnode*[linewidth=0.5pt,fillstyle=solid,fillcolor=white](2,3.7){2pt}{v}
%\cnode*[linewidth=0.5pt,fillstyle=solid,fillcolor=white](2,2.8){2pt}{v}

\rput(2.35,3.7){$p_1$}
%\rput(2.35,2.8){$p_2$}
%\rput(2.35,2.3){$p_3$}
%\rput(2.69,1.4){$p_{\xi(l)-2}$}
%\rput(2.69,0.8){$p_{\xi(l)-1}$}
\rput(2.5,0.3){$p_{\xi(l)}$}

\rput(2,-0.9){$\sigma_1(l)$}
\rput(4,-1.4){$\sigma_4(l)=\sigma_5(l)=\ldots=\emptyset$}

}
\endpspicture
\pspicture(1.6,4)
\scalebox{0.8}{

\psline[linewidth=0.5pt,linecolor=black,linestyle=dashed, dash=2pt 2pt](2,-0.5)(2,4.5)

%\psline[linewidth=1pt,linecolor=black](2,2.3)(2,1.4)
\psline[linewidth=1pt,linecolor=black](2,2.8)(2,0.8)

%\cnode*[linewidth=0.5pt,fillstyle=solid,fillcolor=white](2,0.3){2pt}{v}
\cnode*[linewidth=0.5pt,fillstyle=solid,fillcolor=white](2,.8){2pt}{v}
%\cnode*[linewidth=0.5pt,fillstyle=solid,fillcolor=white](2,1.4){2pt}{v}
%\cnode*[linewidth=0.5pt,fillstyle=solid,fillcolor=white](2,2.3){2pt}{v}
%\cnode*[linewidth=0.5pt,fillstyle=solid,fillcolor=white](2,3.7){2pt}{v}
\cnode*[linewidth=0.5pt,fillstyle=solid,fillcolor=white](2,2.8){2pt}{v}

%\rput(2.35,3.7){$p_1$}
\rput(2.35,2.8){$p_2$}
%\rput(2.35,2.3){$p_3$}
%\rput(2.69,1.4){$p_{\xi(l)-2}$}
\rput(2.69,0.8){$p_{\xi(l)-1}$}
%\rput(2.5,0.3){$p_{\xi(l)}$}

\rput(2,-0.9){$\sigma_2(l)$}

}
\endpspicture
\pspicture(1.6,4)
\scalebox{0.8}{

\psline[linewidth=0.5pt,linecolor=black,linestyle=dashed, dash=2pt 2pt](2,-0.5)(2,4.5)

%\psline[linewidth=1pt,linecolor=black](2,2.3)(2,1.4)
\psline[linewidth=1pt,linecolor=black](2,2.3)(2,1.4)

%\cnode*[linewidth=0.5pt,fillstyle=solid,fillcolor=white](2,0.3){2pt}{v}
%\cnode*[linewidth=0.5pt,fillstyle=solid,fillcolor=white](2,.8){2pt}{v}
\cnode*[linewidth=0.5pt,fillstyle=solid,fillcolor=white](2,1.4){2pt}{v}
\cnode*[linewidth=0.5pt,fillstyle=solid,fillcolor=white](2,2.3){2pt}{v}
%\cnode*[linewidth=0.5pt,fillstyle=solid,fillcolor=white](2,3.7){2pt}{v}
%\cnode*[linewidth=0.5pt,fillstyle=solid,fillcolor=white](2,2.8){2pt}{v}

%\rput(2.35,3.7){$p_1$}
%\rput(2.35,2.8){$p_2$}
\rput(2.35,2.3){$p_3$}
\rput(2.69,1.4){$p_{\xi(l)-2}$}
%\rput(2.69,0.8){$p_{\xi(l)-1}$}
%\rput(2.5,0.3){$p_{\xi(l)}$}

\rput(2,-0.9){$\sigma_3(l)$}

}
\endpspicture

\vspace{1.2cm}
\caption{Definition of the $k$-span.}\label{fig:k-span}
\end{center}
\end{figure}

Let $\xi(l)$ be the number of endpoints along $l$ and let $p_1,\ldots,p_{\xi(l)}$ be the endpoints along $l$, ordered downward. For a positive integer $k$, we define the {\sl $k$-span}, denoted by $\sigma_k(l)$, of $l$ with respect to $W$ as follows (see Fig.~\ref{fig:k-span} for an illustration). If $\xi(l) \le 2(m-1)$, then $\sigma_k(l) \doteq \emptyset$; otherwise, $\sigma_k(l)$ is defined to be the (possibly zero-length) line segment $p_kp_{\xi(l)-k+1}$. Now, we say that $l$ is an {\sl $k$-perfect cut with respect to $W$} if $\sigma_k(l) \subseteq {\rm Un}(G)$, see Fig.~\ref{fig:k-perfect}. Notice that if $\xi(l) \le 2(m-1)$, then $l$ is trivially a $k$-perfect cut as $\sigma_k(l)=\emptyset$, and if $\xi(l) = 2m-1$, then $l$ is $k$-perfect as $\sigma_k(l)$ is a single point.

\begin{figure}
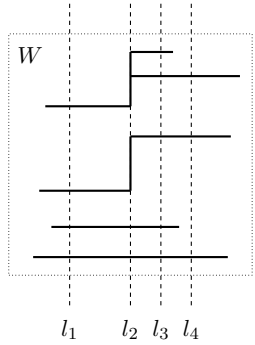

\begin{center}
\pspicture(3.2,3.5)
\scalebox{0.8}{

\psline[linewidth=0.5pt,linestyle=dotted,dotsep=1pt](0,0)(4,0)(4,4)(0,4)(0,0)

\psline[linewidth=0.5pt,linecolor=black,linestyle=dashed, dash=2pt 2pt](2,-0.5)(2,4.5)
\psline[linewidth=0.5pt,linecolor=black,linestyle=dashed, dash=2pt 2pt](1,-0.5)(1,4.5)
\psline[linewidth=0.5pt,linecolor=black,linestyle=dashed, dash=2pt 2pt](2.5,-0.5)(2.5,4.5)
\psline[linewidth=0.5pt,linecolor=black,linestyle=dashed, dash=2pt 2pt](3,-0.5)(3,4.5)

\psline[linewidth=1pt,linecolor=black](0.4,0.3)(3.6,0.3)
\psline[linewidth=1pt,linecolor=black](0.7,0.8)(2.8,0.8)
\psline[linewidth=1pt,linecolor=black](0.5,1.4)(2,1.4)
\psline[linewidth=1pt,linecolor=black](2,2.3)(3.65,2.3)
\psline[linewidth=1pt,linecolor=black](0.6,2.8)(2,2.8)
\psline[linewidth=1pt,linecolor=black](2,3.3)(3.8,3.3)
\psline[linewidth=1pt,linecolor=black](2.7,3.7)(2,3.7)

\psline[linewidth=1pt,linecolor=black](2,2.3)(2,1.4)
\psline[linewidth=1pt,linecolor=black](2,3.7)(2,2.8)

\rput(0.35,3.65){$W$}
\rput(2,-0.9){$l_2$}
\rput(1,-0.9){$l_1$}
\rput(2.5,-0.9){$l_3$}
\rput(3,-0.9){$l_4$}

}
\endpspicture

\vspace{0.8cm}
\caption{The vertical cuts $l_1, l_2, l_3,l_4$ are $3$-perfect (also $k$ perfect, $k \ge 4$).\\ The cut $l_4$
is also $2$-perfect but not $1$-perfect.}\label{fig:k-perfect}
\end{center}
\end{figure}

Following the above definitions, we say that $G$ is a {\sl $k$-guillotine subdivision with respect to a window $W$} if either (1) $V(G) \cap {\rm int}(W) = \emptyset$ or (2) there exists a $k$-perfect cut $l$ with respect to $W$ such that $G$ is a $k$-guillotine subdivision with respect to windows $W \cap H^+$ and $W \cap H^-$, where $H^+$ and $H^-$ are the closed half-planes induced by the cut $l$. Finally, we say that $G$ is a {\sl $k$-guillotine subdivision} if $G$ is a $k$-guillotine subdivision with respect to the unit square $B$.

%In~\cite{M-1997}, Mitchell proved the following key theorem.
\begin{theorem}\label{thm:Mitchell}{\em\cite{M-1997}}
Let $G$ be a rectilinear graph. Then, for any integer $k \ge 1$, there exists a~$k$-guillotine subdivision $G_k$ such that $l(G_k) \le (1+\frac{1}{k})\, l(G)$ and ${\rm Un}(G) \subseteq {\rm Un}(G_k)$.
\end{theorem}

Let $F_{\rm OPT}$ be a minimum rectilinear Steiner forest arborescence  for an $n$-point set $P$ with an $m$-root set $R$, both lying in the first quadrant of $\mathbb{R}^2$ and such that $(0,0) \in R$. By the above theorem,  $F_{\rm OPT}$ can be converted to a~$k$-guillotine subdivision $F_k$ without increasing its length by much and such that ${\rm Un}(G) \subseteq {\rm Un}(G_k)$. In particular, in the resulting subdivision $F_k$, for each point $p \in P$, there exists a root $r \in R$ and a path $\pi$ in $F_k$ such that the length of $\pi$ is equal to $({\rm x}(p)-{\rm x}(r))+({\rm y}(p)-{\rm y}(r))$.
On the other hand, if $G$ is a $k$-guillotine (rectilinear) subdivision such that for each point $p \in P$, there exists a root $r \in R$ and a path $\pi$ in $G$ such that the length of $\pi$ is equal to $({\rm x}(p)-{\rm x}(r))+({\rm y}(p)-{\rm y}(r))$, then $G$ can be converted into a rectilinear Steiner forest arborescence $F_G$ for $P \cup R$ of length $l(F_G) \le l(G)$, in a simple greedy manner, be deleting some edges. Moreover, by minimality of $l(G)$, we have $l(F_G) \le l(F_k)$. Therefore, using the recursive structure of a~$k$-guillotine subdivisions, the idea is to find a minimum-length $k$-guillotine (rectilinear) subdivision $G$ such that for each point $p \in P$, there exists a root $r \in R$ and a path $\pi$ in $G$ such that the length of $\pi$ is equal to $({\rm x}(p)-{\rm x}(r))+({\rm y}(p)-{\rm y}(r))$. Note that an optimal solution will be necessarily a forest, since any cycle is formed can be broken without violating the connectivity requirements
in the objective function (defined below).

\paragraph{A minimum-length subdivision.} We assume, without loss of generality, that no two points of the input sets $P$ and $R$ lie on a common vertical or horizontal line. Otherwise, we slightly perturbate points in $P \cup R$, but in such a way that for a given root $r \in R$, perturbating does not change the subset of points in $P$ that are covered by $r$.

Let $x_1 < x_2 < \ldots < x_{2(n+m)-1}$ (resp.\  $y_1 < y_2 < \ldots < y_{2(n+m)-1}$) denote the sorted $x$- (resp.\ $y$-) coordinates of the $n+m$ points in $P \cup R$ as well as the $n+m-1$ midpoints of the intervals determined by these coordinates. Recall that by~\cite{M-1999}, we may discretise $k$-perfect cuts, that is, we may restrict the $x$-coordinate of a vertical $k$-perfect cut (resp.\ the $y$-coordinate of a horizontal $k$-perfect cut) to be in $\{x_1,\ldots, x_{2(n+m)-1}\}$ (resp.\ in $\{y_1,\ldots, y_{2(n+m)-1}\}$). Observe that by the construction, in particular, by the definition of the $k$-span, the~$k$-guillotine subdivision $G_k$ obtained in Theorem~3.\ref{thm:Mitchell} has the following property: {\sl if $p \in {\rm Un}(G_k)$, then there exists a root $r \in R$ and a path $\pi$ in $G_k$ such that the length of $\pi$ is equal to $({\rm x}(p)-{\rm x}(r))+({\rm y}(p)-{\rm y}(r))$}. This fact is crucial since it allows us to restrict ourselves only to $k$-guillotine subdivisions that satisfy this property, without loosing the estimate on the approximation ratio. Specifically, our dynamic programming is based on solving the following subproblems, see Fig.~\ref{fig:subproblem}. 
%\newpage
\\[5mm]
\hrule
\vspace{3mm}
\noindent{\textbf{\small SUBPROBLEM}}

\vspace{3mm}
\hrule

\vspace{3mm}
\noindent{\bfseries{Instance.}}
\begin{itemize}
\item[(1)] An axis-aligned window $W=W(l,r,b,t)$ determined by two points $(x_l,y_b)$ and $(x_r,y_t)$, $x_l <x_r$ and $y_b < y_t$, with $e_r,e_t,e_l$ and $e_b$ as its right, top, left and bottom edge, respectively.
\item[(2)] For each edge $e$ of $W$, a set $X_e$ of at most $2k+2$ points of $e$, among them at most $2k$ distinct internal points of $e$ and at most two points corresponding to the endpoints of $e$.
\item[(3)] For each set $X_e$, a set $S_e$ of disjoint line segments (vertical or horizontal), called {\sl gates}, with endpoints in $S_e$ such that if a point $p \in P$ lies on the boundary of $W$, then $p$ belongs to some segment in $S_e$; 
in addition, if $X_e$ has $2k$ distinct internal points of $e$, then the segment formed by the middle two of them belongs to $S_e$.
\end{itemize}

\noindent{\bfseries{Objective.}}  Let  $P_W=P \cap {\rm int}(W)$, let $P_r$ be the set of the bottom endpoints of segments in $S_{e_r}$, including the degenerate segments in $S_{e_r}$, and let $P_t$ be the set of the left endpoints of segments in $S_{e_t}$, including the degenerate segments in $S_{e_t}$. Set $I_W \doteq (P_W \cup P_r \cup P_t) \setminus R$. The objective is to compute a minimum-length guillotine rectilinear subdivision $G=(V_G,E_G)$ such that:

\begin{itemize}
\item[(1)] Edges in $E_G$ are axis-aligned line segments not lying on the boundary of $W$.
\item[(2)] For each point $p \in I_W$, there is a path $\pi$ in ${\rm Un}(G)$ to either a root $r \in R \cap {\rm int}(W)$ or to a point $q$ of a segment in $X_{e_l} \cup X_{e_b}$ such that the length of $\pi$ is equal to $({\rm x}(p)-{\rm x}(u))+({\rm y}(p)-{\rm y}(u))$, for some $u \in \{r,s\}$.
\end{itemize}
 \hrule

%\newpage
\bigskip
As regards the condition (1) in the objective function --- by arguments similar to those in the proof of Lemma~1.\ref{lem:preRSA_in_H}, so we omit details --- observe that, without loss of generality, we may assume that the edges in $E_G$ lie also within the Hanan grid induced by all points in $((P \cup R) \cap W) \cup \bigcup_{e \in E(W)}X_e$, where $E(W)=\{e_r,e_t,e_l,e_b\}$.

\begin{figure}
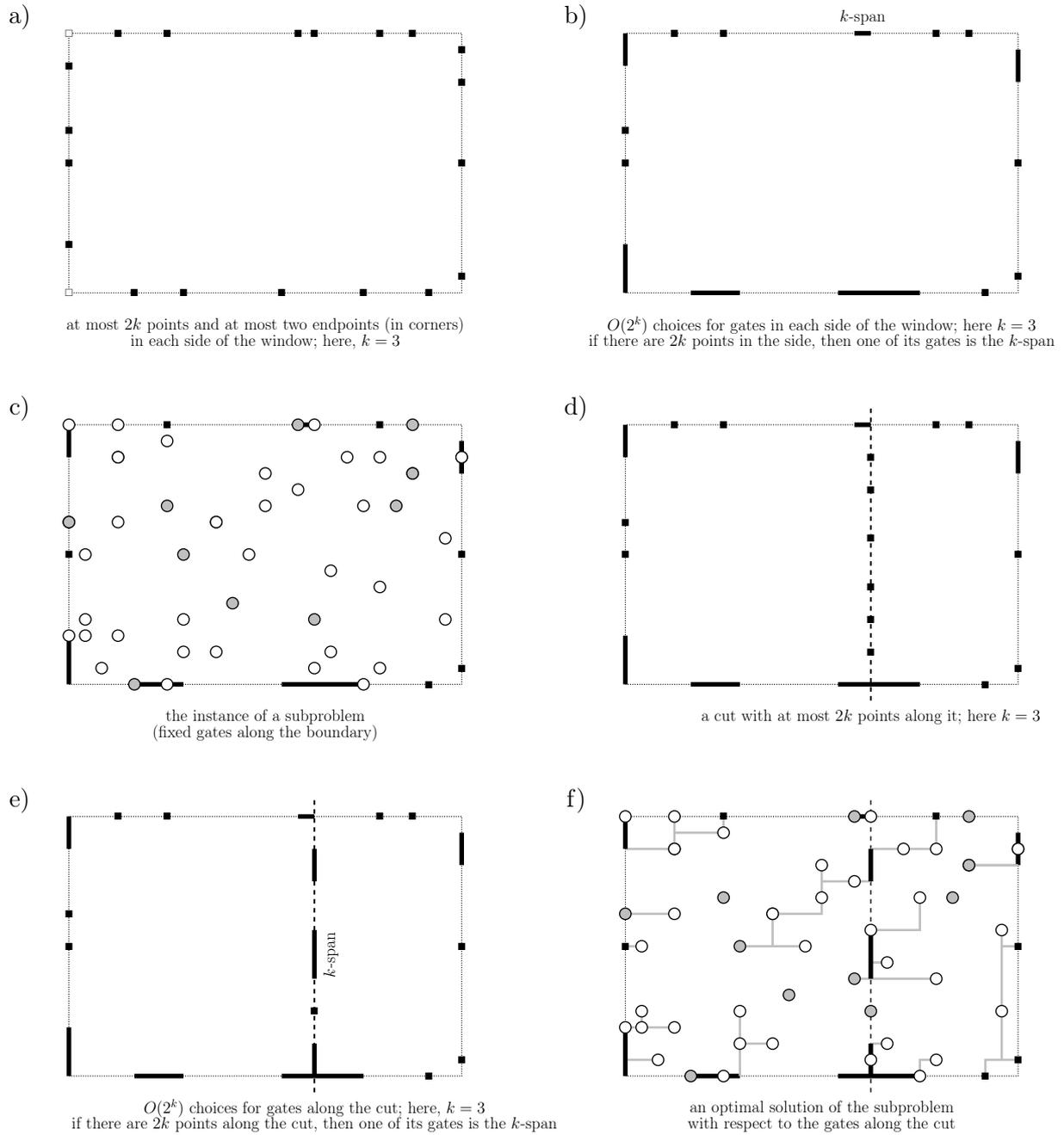

\begin{center}
\pspicture(8.5,4.5)
\scalebox{0.5}{

\psline[linewidth=1pt,linestyle=dotted,dotsep=1pt](0,0)(12,0)(12,8)(0,8)(0,0)

%\psline[linewidth=1pt,linecolor=black,linestyle=dashed, dash=4pt 4pt](7.5,-0.5)(7.5,8.5)

\rput(2,0){\footnotesize$\blacksquare$}
\rput(3.5,0){\footnotesize$\blacksquare$}
\rput(6.5,0){\footnotesize$\blacksquare$}
\rput(9,0){\footnotesize$\blacksquare$}
\rput(0,0){\footnotesize$\white\blacksquare$}
\rput(0,0){\footnotesize$\square$}
\rput(0,1.5){\footnotesize$\blacksquare$}
\rput(0,7){\footnotesize$\blacksquare$}
\rput(0,8){\footnotesize$\white\blacksquare$}
\rput(0,8){\footnotesize$\square$}
\rput(12,6.5){\footnotesize$\blacksquare$}
\rput(12,7.5){\footnotesize$\blacksquare$}
\rput(7,8){\footnotesize$\blacksquare$}
\rput(7.5,8){\footnotesize$\blacksquare$}

\rput(11,0){\footnotesize$\blacksquare$}
\rput(12,0.5){\footnotesize$\blacksquare$}
\rput(12,4){\footnotesize$\blacksquare$}
\rput(10.5,8){\footnotesize$\blacksquare$}
\rput(9.5,8){\footnotesize$\blacksquare$}
\rput(3,8){\footnotesize$\blacksquare$}
\rput(1.5,8){\footnotesize$\blacksquare$}
\rput(0,5){\footnotesize$\blacksquare$}
\rput(0,4){\footnotesize$\blacksquare$}

\rput(6,-1){\Large at most $2k$ points and at most two endpoints (in corners)}
\rput(6,-1.5){\Large in each side of the window; here, $k=3$}

}
\rput(-0.75,4.25){a)}
\endpspicture
\pspicture(6,4.5)
\scalebox{0.5}{

\psline[linewidth=1pt,linestyle=dotted,dotsep=1pt](0,0)(12,0)(12,8)(0,8)(0,0)

\psline[linewidth=4pt,linecolor=black](2,0)(3.5,0)
\psline[linewidth=4pt,linecolor=black](6.5,0)(9,0)

\psline[linewidth=4pt,linecolor=black](0,0)(0,1.5)

\psline[linewidth=4pt,linecolor=black](0,7)(0,8)

\psline[linewidth=4pt,linecolor=black](12,6.5)(12,7.5)
\psline[linewidth=4pt,linecolor=black](7,8)(7.5,8)

\rput(11,0){\footnotesize$\blacksquare$}
\rput(12,0.5){\footnotesize$\blacksquare$}
\rput(12,4){\footnotesize$\blacksquare$}
\rput(10.5,8){\footnotesize$\blacksquare$}
\rput(9.5,8){\footnotesize$\blacksquare$}
\rput(3,8){\footnotesize$\blacksquare$}
\rput(1.5,8){\footnotesize$\blacksquare$}
\rput(0,5){\footnotesize$\blacksquare$}
\rput(0,4){\footnotesize$\blacksquare$}
%\rput(7.5,2){\footnotesize$\blacksquare$}
%\rput(7.5,6){\footnotesize$\blacksquare$}
%\rput(7.5,7){\footnotesize$\blacksquare$}
%\rput(7.5,3){\footnotesize$\blacksquare$}
%\rput(7.5,4.5){\footnotesize$\blacksquare$}

\rput(6,-1){\Large $O(2^k)$ choices for gates in each side of the window; here $k=3$}
\rput(6,-1.5){\Large if there are $2k$ points in the side, then one of its gates is the $k$-span}
\rput(7.25,8.5){\Large $k$-span}

}
\rput(-0.75,4.25){b)}
\endpspicture

\pspicture(8.5,6)
\scalebox{0.5}{

\psline[linewidth=1pt,linestyle=dotted,dotsep=1pt](0,0)(12,0)(12,8)(0,8)(0,0)

\psline[linewidth=4pt,linecolor=black](2,0)(3.5,0)
\cnode[linewidth=1pt,fillstyle=solid,linecolor=black,fillcolor=lightgray](2,0){5pt}{v}
\psline[linewidth=4pt,linecolor=black](6.5,0)(9,0)

\psline[linewidth=4pt,linecolor=black](0,0)(0,1.5)

\cnode[linewidth=1pt,fillstyle=solid,fillcolor=white](0,1.5){5pt}{v}
\cnode[linewidth=1pt,fillstyle=solid,fillcolor=white](0.5,1.5){5pt}{v}
\cnode[linewidth=1pt,fillstyle=solid,fillcolor=white](0.5,2){5pt}{v}
\cnode[linewidth=1pt,fillstyle=solid,fillcolor=white](1.5,1.5){5pt}{v}
\cnode[linewidth=1pt,fillstyle=solid,fillcolor=white](1,.5){5pt}{v}
\cnode[linewidth=1pt,fillstyle=solid,fillcolor=white,linecolor=white](1.5,8){5pt}{v}

\cnode[linewidth=1pt,fillstyle=solid,fillcolor=white](3,0){5pt}{v}
\cnode[linewidth=1pt,fillstyle=solid,fillcolor=white](3.5,1){5pt}{v}
\cnode[linewidth=1pt,fillstyle=solid,fillcolor=white](3.5,2){5pt}{v}
\cnode[linewidth=1pt,fillstyle=solid,fillcolor=white](4.5,1){5pt}{v}

\cnode[linewidth=1pt,fillstyle=solid,fillcolor=white](7.5,.5){5pt}{v}
\cnode[linewidth=1pt,fillstyle=solid,fillcolor=white](8,1){5pt}{v}
\cnode[linewidth=1pt,fillstyle=solid,fillcolor=white](9,0){5pt}{v}
\cnode[linewidth=1pt,fillstyle=solid,fillcolor=white](9.5,.5){5pt}{v}

\cnode[linewidth=1pt,fillstyle=solid,fillcolor=white](11.5,4.5){5pt}{v}

\cnode[linewidth=1pt,fillstyle=solid,fillcolor=white](0.5,4){5pt}{v}

%\cnode*[linewidth=1pt,fillstyle=solid,linecolor=lightgray,fillcolor=white](0,5){5pt}{v}
%\cnode*[linewidth=1pt,fillstyle=solid,linecolor=lightgray,fillcolor=white](0,5){5pt}{v}
\cnode[linewidth=1pt,fillstyle=solid,linecolor=black,fillcolor=lightgray](0,5){5pt}{v}
\cnode[linewidth=1pt,fillstyle=solid,linecolor=black,fillcolor=lightgray](0,5){5pt}{v}
\cnode[linewidth=1pt,fillstyle=solid,fillcolor=white](1.5,5){5pt}{v}

\psline[linewidth=4pt,linecolor=black](0,7)(0,8)

\cnode[linewidth=1pt,fillstyle=solid,fillcolor=white](0,8){5pt}{v}
\cnode[linewidth=1pt,fillstyle=solid,fillcolor=white](1.5,7){5pt}{v}
\cnode[linewidth=1pt,fillstyle=solid,fillcolor=white](1.5,8){5pt}{v}
\cnode[linewidth=1pt,fillstyle=solid,fillcolor=white](3,7.5){5pt}{v}
\cnode[linewidth=1pt,fillstyle=solid,fillcolor=white](1.5,7){5pt}{v}

\cnode[linewidth=1pt,fillstyle=solid,linecolor=black,fillcolor=lightgray](3.5,4){5pt}{v}
\cnode[linewidth=1pt,fillstyle=solid,fillcolor=white](5.5,4){5pt}{v}
\cnode[linewidth=1pt,fillstyle=solid,fillcolor=white](4.5,5){5pt}{v}

\cnode[linewidth=1pt,fillstyle=solid,fillcolor=white](8,3.5){5pt}{v}
\cnode[linewidth=1pt,fillstyle=solid,fillcolor=white](9,5.5){5pt}{v}
\cnode[linewidth=1pt,fillstyle=solid,fillcolor=white](9.5,3){5pt}{v}

\psline[linewidth=4pt,linecolor=black](12,6.5)(12,7.5)
\cnode[linewidth=1pt,fillstyle=solid,linecolor=black,fillcolor=lightgray](10.5,6.5){5pt}{v}
\cnode[linewidth=1pt,fillstyle=solid,fillcolor=white](12,7){5pt}{v}

\cnode[linewidth=1pt,fillstyle=solid,linecolor=black,fillcolor=lightgray](10.5,8){5pt}{v}

\psline[linewidth=4pt,linecolor=black](12,6.5)(12,7.5)
\cnode[linewidth=1pt,fillstyle=solid,linecolor=black,fillcolor=lightgray](10.5,6.5){5pt}{v}
\cnode[linewidth=1pt,fillstyle=solid,fillcolor=white](12,7){5pt}{v}
\cnode[linewidth=1pt,fillstyle=solid,fillcolor=white](8.5,7){5pt}{v}
\cnode[linewidth=1pt,fillstyle=solid,fillcolor=white](9.5,7){5pt}{v}
\cnode[linewidth=1pt,fillstyle=solid,fillcolor=white](7,6){5pt}{v}
\cnode[linewidth=1pt,fillstyle=solid,fillcolor=white](6,5.5){5pt}{v}
\cnode[linewidth=1pt,fillstyle=solid,fillcolor=white](6,6.5){5pt}{v}
\cnode[linewidth=1pt,fillstyle=solid,fillcolor=white](4.5,5){5pt}{v}

\cnode[linewidth=1pt,fillstyle=solid,fillcolor=white](11.5,2){5pt}{v}
\cnode[linewidth=1pt,fillstyle=solid,linecolor=black,fillcolor=lightgray](7.5,2){5pt}{v}

\cnode[linewidth=1pt,fillstyle=solid,linecolor=black,fillcolor=lightgray](5,2.5){5pt}{v}
\cnode[linewidth=1pt,fillstyle=solid,linecolor=black,fillcolor=lightgray](3,5.5){5pt}{v}
\cnode[linewidth=1pt,fillstyle=solid,linecolor=black,fillcolor=lightgray](10,5.5){5pt}{v}

\psline[linewidth=4pt,linecolor=black](7,8)(7.5,8)
\cnode[linewidth=1pt,fillstyle=solid,linecolor=black,fillcolor=lightgray](7,8){5pt}{v}
\cnode[linewidth=1pt,fillstyle=solid,fillcolor=white](7.5,8){5pt}{v}

\rput(11,0){\footnotesize$\blacksquare$}
\rput(12,0.5){\footnotesize$\blacksquare$}
\rput(12,4){\footnotesize$\blacksquare$}
%\rput(10.5,8){\footnotesize$\blacksquare$}
\rput(9.5,8){\footnotesize$\blacksquare$}
\rput(3,8){\footnotesize$\blacksquare$}
%\rput(1.5,8){\footnotesize$\blacksquare$}
%\rput(0,5){\footnotesize$\blacksquare$}
\rput(0,4){\footnotesize$\blacksquare$}
%\rput(7.5,7){\footnotesize$\blacksquare$}
%\rput(7.5,2){\footnotesize$\blacksquare$}

\rput(6,-1){\Large the instance of a subproblem}
\rput(6,-1.5){\Large (fixed gates along the boundary)}

}
\rput(-0.75,4.25){c)}
\endpspicture
\pspicture(6,6)
\scalebox{0.5}{

\psline[linewidth=1pt,linestyle=dotted,dotsep=1pt](0,0)(12,0)(12,8)(0,8)(0,0)

\psline[linewidth=1.5pt,linecolor=black,linestyle=dashed, dash=4pt 4pt](7.5,-0.5)(7.5,8.5)

\psline[linewidth=4pt,linecolor=black](2,0)(3.5,0)
\psline[linewidth=4pt,linecolor=black](6.5,0)(9,0)

\psline[linewidth=4pt,linecolor=black](0,0)(0,1.5)

\psline[linewidth=4pt,linecolor=black](0,7)(0,8)

\psline[linewidth=4pt,linecolor=black](12,6.5)(12,7.5)
\psline[linewidth=4pt,linecolor=black](7,8)(7.5,8)

\rput(11,0){\footnotesize$\blacksquare$}
\rput(12,0.5){\footnotesize$\blacksquare$}
\rput(12,4){\footnotesize$\blacksquare$}
\rput(10.5,8){\footnotesize$\blacksquare$}
\rput(9.5,8){\footnotesize$\blacksquare$}
\rput(3,8){\footnotesize$\blacksquare$}
\rput(1.5,8){\footnotesize$\blacksquare$}
\rput(0,5){\footnotesize$\blacksquare$}
\rput(0,4){\footnotesize$\blacksquare$}
\rput(7.5,2){\footnotesize$\blacksquare$}
\rput(7.5,6){\footnotesize$\blacksquare$}
\rput(7.5,7){\footnotesize$\blacksquare$}
\rput(7.5,3){\footnotesize$\blacksquare$}
\rput(7.5,4.5){\footnotesize$\blacksquare$}
\rput(7.5,1){\footnotesize$\blacksquare$}

\rput(7.5,-1){\Large a cut with at most $2k$ points along it; here $k=3$}

}
\rput(-0.75,4.25){d)}
\endpspicture

\pspicture(8.5,6)
\scalebox{0.5}{

\psline[linewidth=1pt,linestyle=dotted,dotsep=1pt](0,0)(12,0)(12,8)(0,8)(0,0)

\psline[linewidth=1.5pt,linecolor=black,linestyle=dashed, dash=4pt 4pt](7.5,-0.5)(7.5,8.5)

\psline[linewidth=4pt,linecolor=black](2,0)(3.5,0)
\psline[linewidth=4pt,linecolor=black](6.5,0)(9,0)

\psline[linewidth=4pt,linecolor=black](0,0)(0,1.5)

\psline[linewidth=4pt,linecolor=black](0,7)(0,8)

\psline[linewidth=4pt,linecolor=black](12,6.5)(12,7.5)
\psline[linewidth=4pt,linecolor=black](7,8)(7.5,8)

\psline[linewidth=4pt,linecolor=black](7.5,0)(7.5,1)
\psline[linewidth=4pt,linecolor=black](7.5,3)(7.5,4.5)
\psline[linewidth=4pt,linecolor=black](7.5,6)(7.5,7)

\rput(11,0){\footnotesize$\blacksquare$}
\rput(12,0.5){\footnotesize$\blacksquare$}
\rput(12,4){\footnotesize$\blacksquare$}
\rput(10.5,8){\footnotesize$\blacksquare$}
\rput(9.5,8){\footnotesize$\blacksquare$}
\rput(3,8){\footnotesize$\blacksquare$}
\rput(1.5,8){\footnotesize$\blacksquare$}
\rput(0,5){\footnotesize$\blacksquare$}
\rput(0,4){\footnotesize$\blacksquare$}
\rput(7.5,2){\footnotesize$\blacksquare$}

\rput(7.5,-1){\Large $O(2^k)$ choices for gates along the cut; here, $k=3$}
\rput(7.5,-1.5){\Large if there are $2k$ points along the cut, then one of its gates is the $k$-span}

\rput(8,3.75){\rotateleft{\Large $k$-span}}

}
\rput(-0.75,4.25){e)}
\endpspicture
\pspicture(6,6)
\scalebox{0.5}{

\psline[linewidth=1pt,linestyle=dotted,dotsep=1pt](0,0)(12,0)(12,8)(0,8)(0,0)

\psline[linewidth=1pt,linecolor=black,linestyle=dashed, dash=4pt 4pt](7.5,-0.5)(7.5,8.5)

\psline[linewidth=4pt,linecolor=black](2,0)(3.5,0)
\cnode[linewidth=1pt,fillstyle=solid,linecolor=black,fillcolor=lightgray](2,0){5pt}{v}
\psline[linewidth=4pt,linecolor=black](6.5,0)(9,0)

\psline[linewidth=2pt,linecolor=lightgray](0,0.5)(1,0.5)
\psline[linewidth=2pt,linecolor=lightgray](0.5,1.5)(0.5,2)
\psline[linewidth=2pt,linecolor=lightgray](0,1.5)(1.5,1.5)
\psline[linewidth=4pt,linecolor=black](0,0)(0,1.5)

\cnode[linewidth=1pt,fillstyle=solid,fillcolor=white](0,1.5){5pt}{v}
\cnode[linewidth=1pt,fillstyle=solid,fillcolor=white](0.5,1.5){5pt}{v}
\cnode[linewidth=1pt,fillstyle=solid,fillcolor=white](0.5,2){5pt}{v}
\cnode[linewidth=1pt,fillstyle=solid,fillcolor=white](1.5,1.5){5pt}{v}
\cnode[linewidth=1pt,fillstyle=solid,fillcolor=white](1,.5){5pt}{v}
\cnode[linewidth=1pt,fillstyle=solid,fillcolor=white,linecolor=white](1.5,8){5pt}{v}

\psline[linewidth=2pt,linecolor=lightgray](3.5,0)(3.5,2)
\psline[linewidth=2pt,linecolor=lightgray](3.5,1)(4.5,1)

\cnode[linewidth=1pt,fillstyle=solid,fillcolor=white](3,0){5pt}{v}
\cnode[linewidth=1pt,fillstyle=solid,fillcolor=white](3.5,1){5pt}{v}
\cnode[linewidth=1pt,fillstyle=solid,fillcolor=white](3.5,2){5pt}{v}
\cnode[linewidth=1pt,fillstyle=solid,fillcolor=white](4.5,1){5pt}{v}

\psline[linewidth=2pt,linecolor=lightgray](7.5,1)(8,1)
\psline[linewidth=2pt,linecolor=lightgray](9,0)(9,0.5)(9.5,0.5)
\psline[linewidth=4pt,linecolor=black](7.5,0)(7.5,1)

\cnode[linewidth=1pt,fillstyle=solid,fillcolor=white](7.5,.5){5pt}{v}
\cnode[linewidth=1pt,fillstyle=solid,fillcolor=white](8,1){5pt}{v}
\cnode[linewidth=1pt,fillstyle=solid,fillcolor=white](9,0){5pt}{v}
\cnode[linewidth=1pt,fillstyle=solid,fillcolor=white](9.5,.5){5pt}{v}

\psline[linewidth=2pt,linecolor=lightgray](11,0)(11,0.5)(12,0.5)
\psline[linewidth=2pt,linecolor=lightgray](11.5,0.5)(11.5,4.5)
\psline[linewidth=2pt,linecolor=lightgray](11.5,4)(12,4)
\cnode[linewidth=1pt,fillstyle=solid,fillcolor=white](11.5,4.5){5pt}{v}

\psline[linewidth=2pt,linecolor=lightgray](0,4)(0.5,4)
\cnode[linewidth=1pt,fillstyle=solid,fillcolor=white](0.5,4){5pt}{v}

\psline[linewidth=2pt,linecolor=lightgray](0,5)(1.5,5)
\cnode[linewidth=1pt,fillstyle=solid,linecolor=black,fillcolor=lightgray](0,5){5pt}{v}
\cnode[linewidth=1pt,fillstyle=solid,fillcolor=white](1.5,5){5pt}{v}

\psline[linewidth=2pt,linecolor=lightgray](0,7)(1.5,7)(1.5,8)
\psline[linewidth=2pt,linecolor=lightgray](1.5,7.5)(3,7.5)(3,8)

\psline[linewidth=4pt,linecolor=black](0,7)(0,8)

\cnode[linewidth=1pt,fillstyle=solid,fillcolor=white](0,8){5pt}{v}
\cnode[linewidth=1pt,fillstyle=solid,fillcolor=white](1.5,7){5pt}{v}
\cnode[linewidth=1pt,fillstyle=solid,fillcolor=white](1.5,8){5pt}{v}
\cnode[linewidth=1pt,fillstyle=solid,fillcolor=white](3,7.5){5pt}{v}
\cnode[linewidth=1pt,fillstyle=solid,fillcolor=white](1.5,7){5pt}{v}

\psline[linewidth=2pt,linecolor=lightgray](3.5,4)(5.5,4)
\psline[linewidth=2pt,linecolor=lightgray](4.5,4)(4.5,5)
\cnode[linewidth=1pt,fillstyle=solid,linecolor=black,fillcolor=lightgray](3.5,4){5pt}{v}
\cnode[linewidth=1pt,fillstyle=solid,fillcolor=white](5.5,4){5pt}{v}
\cnode[linewidth=1pt,fillstyle=solid,fillcolor=white](4.5,5){5pt}{v}

\psline[linewidth=2pt,linecolor=lightgray](7,3)(9.5,3)
\psline[linewidth=2pt,linecolor=lightgray](7.5,3.5)(8,3.5)
\psline[linewidth=2pt,linecolor=lightgray](7.5,4.5)(9,4.5)(9,5.5)
\psline[linewidth=4pt,linecolor=black](7.5,3)(7.5,4.5)
\cnode[linewidth=1pt,fillstyle=solid,linecolor=black,fillcolor=lightgray](7,3){5pt}{v}
\cnode[linewidth=1pt,fillstyle=solid,fillcolor=white](7.5,4.5){5pt}{v}
\cnode[linewidth=1pt,fillstyle=solid,fillcolor=white](8,3.5){5pt}{v}
\cnode[linewidth=1pt,fillstyle=solid,fillcolor=white](9,5.5){5pt}{v}
\cnode[linewidth=1pt,fillstyle=solid,fillcolor=white](9.5,3){5pt}{v}

\psline[linewidth=2pt,linecolor=lightgray](10.5,6.5)(12,6.5)
\psline[linewidth=4pt,linecolor=black](12,6.5)(12,7.5)
\cnode[linewidth=1pt,fillstyle=solid,linecolor=black,fillcolor=lightgray](10.5,6.5){5pt}{v}
\cnode[linewidth=1pt,fillstyle=solid,fillcolor=white](12,7){5pt}{v}

\cnode[linewidth=1pt,fillstyle=solid,linecolor=black,fillcolor=lightgray](10.5,8){5pt}{v}

\psline[linewidth=2pt,linecolor=lightgray](9.5,8)(9.5,7)(8.5,7)(7.5,7)
\psline[linewidth=2pt,linecolor=lightgray](10.5,6.5)(12,6.5)
\psline[linewidth=2pt,linecolor=lightgray](7.5,6)(6,6)
\psline[linewidth=2pt,linecolor=lightgray](6,6.5)(6,5)(4.5,5)

\psline[linewidth=4pt,linecolor=black](7.5,6)(7.5,7)
\cnode[linewidth=1pt,fillstyle=solid,linecolor=black,fillcolor=lightgray](10.5,6.5){5pt}{v}
\cnode[linewidth=1pt,fillstyle=solid,fillcolor=white](12,7){5pt}{v}
\cnode[linewidth=1pt,fillstyle=solid,fillcolor=white](8.5,7){5pt}{v}
\cnode[linewidth=1pt,fillstyle=solid,fillcolor=white](9.5,7){5pt}{v}
\cnode[linewidth=1pt,fillstyle=solid,fillcolor=white](7,6){5pt}{v}
\cnode[linewidth=1pt,fillstyle=solid,fillcolor=white](6,5.5){5pt}{v}
\cnode[linewidth=1pt,fillstyle=solid,fillcolor=white](6,6.5){5pt}{v}
\cnode[linewidth=1pt,fillstyle=solid,fillcolor=white](4.5,5){5pt}{v}

\cnode[linewidth=1pt,fillstyle=solid,fillcolor=white](11.5,2){5pt}{v}
\cnode[linewidth=1pt,fillstyle=solid,linecolor=black,fillcolor=lightgray](7.5,2){5pt}{v}

\cnode[linewidth=1pt,fillstyle=solid,linecolor=black,fillcolor=lightgray](5,2.5){5pt}{v}
\cnode[linewidth=1pt,fillstyle=solid,linecolor=black,fillcolor=lightgray](3,5.5){5pt}{v}
\cnode[linewidth=1pt,fillstyle=solid,linecolor=black,fillcolor=lightgray](10,5.5){5pt}{v}

\psline[linewidth=4pt,linecolor=black](7,8)(7.5,8)
\cnode[linewidth=1pt,fillstyle=solid,linecolor=black,fillcolor=lightgray](7,8){5pt}{v}
\cnode[linewidth=1pt,fillstyle=solid,fillcolor=white](7.5,8){5pt}{v}

\rput(11,0){\footnotesize$\blacksquare$}
\rput(12,0.5){\footnotesize$\blacksquare$}
\rput(12,4){\footnotesize$\blacksquare$}
%\rput(10.5,8){\footnotesize$\blacksquare$}
\rput(9.5,8){\footnotesize$\blacksquare$}
\rput(3,8){\footnotesize$\blacksquare$}
%\rput(1.5,8){\footnotesize$\blacksquare$}
%\rput(0,5){\footnotesize$\blacksquare$}
\rput(0,4){\footnotesize$\blacksquare$}
%\rput(7.5,7){\footnotesize$\blacksquare$}
%\rput(7.5,2){\footnotesize$\blacksquare$}

\rput(6,-1){\Large{an optimal solution of the subproblem}}
\rput(6,-1.5){\Large{with respect to the gates along the cut}}

}
\rput(-0.75,4.25){f)}
\endpspicture
\vspace{1.2cm}
\caption{An illustration of the subproblem and the recursive approach.}\label{fig:subproblem}
\end{center}
\end{figure}

The initial problem is specified by the bounding box $B$ such that $P \cup R \cup {\rm int}(B)$ (and so $X_e=S_e=\emptyset$ for each side of $B$). The output is a $k$-guillotine $G$ such that if $p \in {\rm Un}(G_k)$, then there exists a root $r \in R$ and a path $\pi$ in $G_k$ such that the length of $\pi$ is equal to $({\rm x}(p)-{\rm x}(r))+({\rm y}(p)-{\rm y}(r))$ (and so, in particular, for a point $p \in P$, then there exists a root $r \in R$ and a path $\pi$ in $G_k$ such that the length of $\pi$ is equal to $({\rm x}(p)-{\rm x}(r))+({\rm y}(p)-{\rm y}(r))$).

The base subproblem is specified by (a) a window $W$ containing inside no point in $P$, and (b) a constant number of gates: at most $2k+1$ of them per each side of $W$ (see Fig.~\ref{fig:subproblem}). Thus it can be also solved in $(k+m)^3\, 2^{O(k)}$ time by using the exact algorithm described in Section~\ref{sec:Exact}. For all other subproblems, we can find the optimal solution recursively, by joining the solutions of the two obtained by splitting the problem, and optimizing over all choices associated with a cut (including the {\sl cost} of the cut, that is, the sum of the lengths of its gates):
\begin{itemize}
\item[(a)] $O(n+m)$ choices of a cut by a horizontal or vertical line (determined by some $x_i$ or $y_j$);
\item[(b)] $O(2^{4k} \cdot (n+m)^{2k})$ choices of a set of gates on the cut within the window: $O(2^{2k} \cdot (n+m)^{2k})$ choices of internal points on the cut within the window, each of them generating $O(2^{2k})$ sets of gates.
\end{itemize}
Therefore, each subproblem takes $O(2^{4k} \cdot (n+m)^{2k+1})$ time. Since each subproblem, i.e., the sides of its window, inherits the structure of a~cut, the total number of subproblems is $O(2^{16k} \cdot (n+m)^{8k+4})$: there are $O(n+m)^4$ windows and $O(2^{4k} \cdot (n+m)^{2k})$ choices of a set of gates on each of the four sides of a window. Consequently, we obtain an overall time complexity of $O(2^{20k} \cdot (n+m)^{10k+5})=O((n+m)^{10k+5})$.

\begin{theorem}\label{thm:PTAS}
Given any fixed positive integer $k$, any set $P$ of $n$ points, and any set $R$ of $m$ roots, all lying in the first quadrant of the plane,  there is an $O((n+m)^{10k+5})$-time algorithm to compute an approximate rectilinear Steiner forest arborescence for $P$ with the root set $R$ whose length is within factor $(1+\frac{1}{k})$ of optimal.
\end{theorem}

\section{Fixed-parameter algorithm}\label{sec:FPT}

In this section, following~\cite{BTW-00, CC-2018, KS-2011}, we show that the RSFA problem is fixed-parameter tractable with respect to the parameter $h$, the minimum number of horizontal lines that all points are lying on. In particular, we show that the RSFA problem can be solved in $O^\ast(2^h)$ time (here, the $O^\ast$-notation neglects polynomial factors). Recall that the idea behind fixed-parameter tractability is to separate out the complexity into two pieces --- some piece that depends purely and polynomially (if possible) on the size of the input, and some piece that depends on some ``parameter'' to the problem.

Let $P$ and $R$ be two disjoint sets of $n$ points and $m$ roots, respectively,
all lying in the first quadrant of $\mathbb{R}^2$ (with $(0,0) \in R$), and let $H=H(P \cup R)$ be the Hanan grid of $P\cup R$; recall that each point in $P \cup R$ corresponds to some intersection grid point of $H$ and for any $X \subseteq P$, there exists an optimal RSFA for $X$ with the root set $R$ such that it uses only subsegments of the Hanan grid $H(X \cup R) \subseteq H(P \cup R)$ (see Lemma~1.\ref{lem:Hanan}). Assume without loss of generality that all intersection points of $H(P \cup R)$ are located on $v$ vertical lines and $h$ horizontal lines, with $2 \le h \le v$, and let $C=\{c_1,c_2,\ldots, c_{vh}\}$ be the set of interection points of $H(P \cup R)$ sorted lexicographically.

Let $k \in \{1,2, \ldots, vh\}$ be an integer and set $j:=\max(1,k-h+1)$. Now, let $$f_k \colon \{c_j,c_{j+1}\ldots, c_k\} \rightarrow \{0,1\}$$ be a function --- which we refer to as {\sl $k$-eligible} --- such that $$(P \cup R) \cap \{c_j,c_{j+1},\ldots,c_k\} \subseteq f_k^{-1}(1),$$ and let $F_{f_k}$ denote an optimal rectilinear Steiner forest arborescence for $(P \cap \{c_1,c_2,\ldots, c_k\}) \cup f_k^{-1}(1)$ with the root set $R$. The idea is that for each $k$-eligible function $f_k$, the RSFA $F_{f_k}$ candidates to be a partial solution for the global optimum (for $P$), where intersection points in $f_k^{-1}(1)$ play a role of gates, that is, the only points that a partial solution may expand through towards a global solution, see Fig.~\ref{fig:FPT-gates}(a-d) for an illustration. Clearly, in order to solve the RFSA problem for $P$ with the root set $R$, we need to compute $F_{f_k}$ for $k=vh$ and the function $f_k$ such that $$f_k^{-1}(1) = (P \cup R) \cap \{c_{k-h+1},c_{k-h+2}\ldots,c_k\}.$$

\begin{figure}
\begin{center}
\pspicture(6,5)
\rput(-.5,4.4){\small{a)}}

\scalebox{0.8}{
%\psgrid

\psline[linewidth=0.25pt,linecolor=lightgray](0,0)(6,0)
\psline[linewidth=0.25pt,linecolor=lightgray](0,1)(6,1)
\psline[linewidth=0.25pt,linecolor=lightgray](0,2)(6,2)
\psline[linewidth=0.25pt,linecolor=lightgray](0,3)(6,3)
\psline[linewidth=0.25pt,linecolor=lightgray](0,4)(6,4)
\psline[linewidth=0.25pt,linecolor=lightgray](0,5)(6,5)
\psline[linewidth=0.25pt,linecolor=lightgray](0,6)(6,6)

\psline[linewidth=0.25pt,linecolor=lightgray](0,0)(0,6)
\psline[linewidth=0.25pt,linecolor=lightgray](1,0)(1,6)
\psline[linewidth=0.25pt,linecolor=lightgray](2,0)(2,6)
\psline[linewidth=0.25pt,linecolor=lightgray](3,0)(3,6)
\psline[linewidth=0.25pt,linecolor=lightgray](4,0)(4,6)
\psline[linewidth=0.25pt,linecolor=lightgray](5,0)(5,6)
\psline[linewidth=0.25pt,linecolor=lightgray](6,0)(6,6)

\psline[linewidth=2pt,linecolor=gray](0,0)(2,0)
\psline[linewidth=2pt,linecolor=gray](1,0)(1,1)(6,1)

\psline[linewidth=2pt,linecolor=gray](0,3)(1,3)

\psline[linewidth=2pt,linecolor=gray](3,6)(3,5)

\psline[linewidth=2pt,linecolor=gray](3,3)(4,3)(4,5)

\psline[linewidth=2pt,linecolor=gray](4,4)(6,4)(6,6)

\cnode[linewidth=0.5pt,fillstyle=solid,fillcolor=lightgray](0,0){3pt}{v}
\cnode[linewidth=0.5pt,fillstyle=solid,fillcolor=lightgray](0,3){3pt}{v}
\cnode[linewidth=0.5pt,fillstyle=solid,fillcolor=lightgray](3,3){3pt}{v}
\cnode[linewidth=0.5pt,fillstyle=solid,fillcolor=lightgray](3,5){3pt}{v}
\cnode[linewidth=0.5pt,fillstyle=solid,fillcolor=white](2,0){3pt}{v}
\cnode[linewidth=0.5pt,fillstyle=solid,fillcolor=white](1,1){3pt}{v}
\cnode[linewidth=0.5pt,fillstyle=solid,fillcolor=white](6,1){3pt}{v}
\cnode[linewidth=0.5pt,fillstyle=solid,fillcolor=white](1,3){3pt}{v}
\cnode[linewidth=0.5pt,fillstyle=solid,fillcolor=white](4,3){3pt}{v}
\cnode[linewidth=0.5pt,fillstyle=solid,fillcolor=white](3,6){3pt}{v}
\cnode[linewidth=0.5pt,fillstyle=solid,fillcolor=white](5,4){3pt}{v}
\cnode[linewidth=0.5pt,fillstyle=solid,fillcolor=white](4,5){3pt}{v}
%\cnode[linewidth=0.5pt,fillstyle=solid,fillcolor=white](4,7){3pt}{v}
\cnode[linewidth=0.5pt,fillstyle=solid,fillcolor=white](6,6){3pt}{v}
\cnode[linewidth=0.5pt,fillstyle=solid,fillcolor=lightgray](1,6){3pt}{v}

%\rput(0.5,-.75){\small{$x+y=0$}}

\rput(5.8,1.2){\small{$p$}}
\rput(0.2,.2){\small{$r$}}

}
\endpspicture
\pspicture(4.8,6)
\rput(-.5,4.4){\small{b)}}

\scalebox{0.8}{
%\psgrid

\psline[linewidth=0.25pt,linecolor=lightgray](0,0)(6,0)
\psline[linewidth=0.25pt,linecolor=lightgray](0,1)(6,1)
\psline[linewidth=0.25pt,linecolor=lightgray](0,2)(6,2)
\psline[linewidth=0.25pt,linecolor=lightgray](0,3)(6,3)
\psline[linewidth=0.25pt,linecolor=lightgray](0,4)(6,4)
\psline[linewidth=0.25pt,linecolor=lightgray](0,5)(6,5)
\psline[linewidth=0.25pt,linecolor=lightgray](0,6)(6,6)

\psline[linewidth=0.25pt,linecolor=lightgray](0,0)(0,6)
\psline[linewidth=0.25pt,linecolor=lightgray](1,0)(1,6)
\psline[linewidth=0.25pt,linecolor=lightgray](2,0)(2,6)
\psline[linewidth=0.25pt,linecolor=lightgray](3,0)(3,6)
\psline[linewidth=0.25pt,linecolor=lightgray](4,0)(4,6)
\psline[linewidth=0.25pt,linecolor=lightgray](5,0)(5,6)
\psline[linewidth=0.25pt,linecolor=lightgray](6,0)(6,6)

\cnode*[linewidth=0.5pt,fillstyle=solid,fillcolor=lightgray](4,2){2pt}{v}
\cnode*[linewidth=0.5pt,fillstyle=solid,fillcolor=lightgray](4,1){2pt}{v}
\cnode*[linewidth=0.5pt,fillstyle=solid,fillcolor=lightgray](4,0){2pt}{v}
\cnode*[linewidth=0.5pt,fillstyle=solid,fillcolor=lightgray](3,6){2pt}{v}
\cnode*[linewidth=0.5pt,fillstyle=solid,fillcolor=lightgray](3,5){2pt}{v}
\cnode*[linewidth=0.5pt,fillstyle=solid,fillcolor=lightgray](3,4){2pt}{v}
\cnode*[linewidth=0.5pt,fillstyle=solid,fillcolor=lightgray](4,3){2pt}{v}

\rput(3.65,0.2){\small{$c_{29}$}}
\rput(3.65,1.2){\small{$c_{30}$}}
\rput(3.65,2.2){\small{$c_{31}$}}
\rput(3.65,3.2){\small{$c_{32}$}}

\rput(2.65,5.8){\small{$c_{28}$}}
\rput(2.65,4.8){\small{$c_{27}$}}
\rput(2.65,3.8){\small{$c_{26}$}}

\rput(4.25,0.2){\small{$0$}}
\rput(4.25,1.2){\small{$1$}}
\rput(4.25,2.2){\small{$0$}}
\rput(4.25,3.2){\small{$0$}}

\rput(3.25,5.8){\small{$1$}}
\rput(3.25,4.8){\small{$0$}}
\rput(3.25,3.8){\small{$0$}}

}
\endpspicture

\pspicture(6,6)
\rput(-.5,4.4){\small{c)}}

\scalebox{0.8}{
%\psgrid

\psline[linewidth=0.25pt,linecolor=lightgray](0,0)(6,0)
\psline[linewidth=0.25pt,linecolor=lightgray](0,1)(6,1)
\psline[linewidth=0.25pt,linecolor=lightgray](0,2)(6,2)
\psline[linewidth=0.25pt,linecolor=lightgray](0,3)(6,3)
\psline[linewidth=0.25pt,linecolor=lightgray](0,4)(6,4)
\psline[linewidth=0.25pt,linecolor=lightgray](0,5)(6,5)
\psline[linewidth=0.25pt,linecolor=lightgray](0,6)(6,6)

\psline[linewidth=0.25pt,linecolor=lightgray](0,0)(0,6)
\psline[linewidth=0.25pt,linecolor=lightgray](1,0)(1,6)
\psline[linewidth=0.25pt,linecolor=lightgray](2,0)(2,6)
\psline[linewidth=0.25pt,linecolor=lightgray](3,0)(3,6)
\psline[linewidth=0.25pt,linecolor=lightgray](4,0)(4,6)
\psline[linewidth=0.25pt,linecolor=lightgray](5,0)(5,6)
\psline[linewidth=0.25pt,linecolor=lightgray](6,0)(6,6)

\cnode*[linewidth=0.5pt,fillstyle=solid,fillcolor=lightgray](4,2){2pt}{v}
\cnode*[linewidth=0.5pt,fillstyle=solid,fillcolor=lightgray](4,1){2pt}{v}
\cnode*[linewidth=0.5pt,fillstyle=solid,fillcolor=lightgray](4,0){2pt}{v}
\cnode*[linewidth=0.5pt,fillstyle=solid,fillcolor=lightgray](3,6){2pt}{v}
\cnode*[linewidth=0.5pt,fillstyle=solid,fillcolor=lightgray](3,5){2pt}{v}
\cnode*[linewidth=0.5pt,fillstyle=solid,fillcolor=lightgray](3,4){2pt}{v}
\cnode*[linewidth=0.5pt,fillstyle=solid,fillcolor=lightgray](4,3){2pt}{v}

\rput(3.65,0.2){\small{$c_{29}$}}
\rput(3.65,1.2){\small{$c_{30}$}}
\rput(3.65,2.2){\small{$c_{31}$}}
\rput(3.65,3.2){\small{$c_{32}$}}

\rput(2.65,5.8){\small{$c_{28}$}}
\rput(2.65,4.8){\small{$c_{27}$}}
\rput(2.65,3.8){\small{$c_{26}$}}

\rput(4.25,0.2){\small{$0$}}
\rput(4.25,1.2){\small{$0$}}
\rput(4.25,2.2){\small{$0$}}
\rput(4.25,3.2){\small{$1$}}

\rput(3.25,5.8){\small{$1$}}
\rput(3.25,4.8){\small{$1$}}
\rput(3.25,3.8){\small{$0$}}

}
\endpspicture
\pspicture(4.8,5)
\rput(-.5,4.4){\small{d)}}

\scalebox{0.8}{
%\psgrid

\psline[linewidth=0.25pt,linecolor=lightgray](0,0)(6,0)
\psline[linewidth=0.25pt,linecolor=lightgray](0,1)(6,1)
\psline[linewidth=0.25pt,linecolor=lightgray](0,2)(6,2)
\psline[linewidth=0.25pt,linecolor=lightgray](0,3)(6,3)
\psline[linewidth=0.25pt,linecolor=lightgray](0,4)(6,4)
\psline[linewidth=0.25pt,linecolor=lightgray](0,5)(6,5)
\psline[linewidth=0.25pt,linecolor=lightgray](0,6)(6,6)

\psline[linewidth=0.25pt,linecolor=lightgray](0,0)(0,6)
\psline[linewidth=0.25pt,linecolor=lightgray](1,0)(1,6)
\psline[linewidth=0.25pt,linecolor=lightgray](2,0)(2,6)
\psline[linewidth=0.25pt,linecolor=lightgray](3,0)(3,6)
\psline[linewidth=0.25pt,linecolor=lightgray](4,0)(4,6)
\psline[linewidth=0.25pt,linecolor=lightgray](5,0)(5,6)
\psline[linewidth=0.25pt,linecolor=lightgray](6,0)(6,6)

\cnode*[linewidth=0.5pt,fillstyle=solid,fillcolor=lightgray](4,2){2pt}{v}
\cnode*[linewidth=0.5pt,fillstyle=solid,fillcolor=lightgray](4,1){2pt}{v}
\cnode*[linewidth=0.5pt,fillstyle=solid,fillcolor=lightgray](4,0){2pt}{v}
\cnode*[linewidth=0.5pt,fillstyle=solid,fillcolor=lightgray](3,6){2pt}{v}
\cnode*[linewidth=0.5pt,fillstyle=solid,fillcolor=lightgray](3,5){2pt}{v}
\cnode*[linewidth=0.5pt,fillstyle=solid,fillcolor=lightgray](3,4){2pt}{v}
\cnode*[linewidth=0.5pt,fillstyle=solid,fillcolor=lightgray](4,3){2pt}{v}

\rput(3.65,0.2){\small{$c_{29}$}}
\rput(3.65,1.2){\small{$c_{30}$}}
\rput(3.65,2.2){\small{$c_{31}$}}
\rput(3.65,3.2){\small{$c_{32}$}}

\rput(2.65,5.8){\small{$c_{28}$}}
\rput(2.65,4.8){\small{$c_{27}$}}
\rput(2.65,3.8){\small{$c_{26}$}}

\rput(4.25,0.2){\small{$0$}}
\rput(4.25,1.2){\small{$1$}}
\rput(4.25,2.2){\small{$0$}}
\rput(4.25,3.2){\small{$1$}}

\rput(3.25,5.8){\small{$1$}}
\rput(3.25,4.8){\small{$1$}}
\rput(3.25,3.8){\small{$0$}}

}
\endpspicture

\pspicture(6,5)
\rput(-.5,4.4){\small{e)}}

\scalebox{0.8}{
%\psgrid

\pspolygon[linecolor=gray9,fillstyle=solid,fillcolor=gray9](0,0)(0,6)(4,6)(4,4)(5,4)(5,0)

\pspolygon[linecolor=gray8,fillstyle=solid,fillcolor=gray8](0,0)(0,6)(4,6)(4,3)(5,3)(5,0)

\psline[linewidth=0.25pt,linecolor=lightgray](0,0)(6,0)
\psline[linewidth=0.25pt,linecolor=lightgray](0,1)(6,1)
\psline[linewidth=0.25pt,linecolor=lightgray](0,2)(6,2)
\psline[linewidth=0.25pt,linecolor=lightgray](0,3)(6,3)
\psline[linewidth=0.25pt,linecolor=lightgray](0,4)(6,4)
\psline[linewidth=0.25pt,linecolor=lightgray](0,5)(6,5)
\psline[linewidth=0.25pt,linecolor=lightgray](0,6)(6,6)

\psline[linewidth=0.25pt,linecolor=lightgray](0,0)(0,6)
\psline[linewidth=0.25pt,linecolor=lightgray](1,0)(1,6)
\psline[linewidth=0.25pt,linecolor=lightgray](2,0)(2,6)
\psline[linewidth=0.25pt,linecolor=lightgray](3,0)(3,6)
\psline[linewidth=0.25pt,linecolor=lightgray](4,0)(4,6)
\psline[linewidth=0.25pt,linecolor=lightgray](5,0)(5,6)
\psline[linewidth=0.25pt,linecolor=lightgray](6,0)(6,6)

\psline[linewidth=2.5pt,linecolor=gray,linestyle=dotted,dotsep=1pt](4,4)(5,4)

\cnode[linewidth=0.5pt,fillstyle=solid,fillcolor=white](5,4){3pt}{v}

\cnode*[linewidth=0.5pt,fillstyle=solid,fillcolor=lightgray](5,2){2pt}{v}
\cnode*[linewidth=0.5pt,fillstyle=solid,fillcolor=lightgray](5,1){2pt}{v}
\cnode*[linewidth=0.5pt,fillstyle=solid,fillcolor=lightgray](5,0){2pt}{v}
\cnode*[linewidth=0.5pt,fillstyle=solid,fillcolor=lightgray](4,6){2pt}{v}
\cnode*[linewidth=0.5pt,fillstyle=solid,fillcolor=lightgray](4,5){2pt}{v}
\cnode*[linewidth=0.5pt,fillstyle=solid,fillcolor=lightgray](4,4){2pt}{v}
\cnode*[linewidth=0.5pt,fillstyle=solid,fillcolor=lightgray](5,3){2pt}{v}
\cnode*[linewidth=0.5pt,fillstyle=solid,fillcolor=lightgray](5,4){2pt}{v}

\rput(4.65,0.2){\small{$c_{36}$}}
\rput(4.65,1.2){\small{$c_{37}$}}
\rput(4.65,2.2){\small{$c_{38}$}}
\rput(4.65,3.2){\small{$c_{39}$}}

\rput(3.65,5.8){\small{$c_{35}$}}
\rput(3.65,4.8){\small{$c_{32}$}}
\rput(3.65,3.8){\small{$c_{31}$}}

\rput(4.65,4.2){\small{$c_{40}$}}

\rput(5.25,4.2){\small{$1$}}
\rput(4.25,3.8){\small{$1$}}
\rput(5.25,3.2){\small{$?$}}

}
\endpspicture
\pspicture(4.8,6)
\rput(-.5,4.4){\small{f)}}

\scalebox{0.8}{
\pspolygon[linecolor=gray9,fillstyle=solid,fillcolor=gray9](0,0)(0,6)(4,6)(4,4)(5,4)(5,0)

\pspolygon[linecolor=gray8,fillstyle=solid,fillcolor=gray8](0,0)(0,6)(4,6)(4,3)(5,3)(5,0)

\psline[linewidth=0.25pt,linecolor=lightgray](0,0)(6,0)
\psline[linewidth=0.25pt,linecolor=lightgray](0,1)(6,1)
\psline[linewidth=0.25pt,linecolor=lightgray](0,2)(6,2)
\psline[linewidth=0.25pt,linecolor=lightgray](0,3)(6,3)
\psline[linewidth=0.25pt,linecolor=lightgray](0,4)(6,4)
\psline[linewidth=0.25pt,linecolor=lightgray](0,5)(6,5)
\psline[linewidth=0.25pt,linecolor=lightgray](0,6)(6,6)

\psline[linewidth=0.25pt,linecolor=lightgray](0,0)(0,6)
\psline[linewidth=0.25pt,linecolor=lightgray](1,0)(1,6)
\psline[linewidth=0.25pt,linecolor=lightgray](2,0)(2,6)
\psline[linewidth=0.25pt,linecolor=lightgray](3,0)(3,6)
\psline[linewidth=0.25pt,linecolor=lightgray](4,0)(4,6)
\psline[linewidth=0.25pt,linecolor=lightgray](5,0)(5,6)
\psline[linewidth=0.25pt,linecolor=lightgray](6,0)(6,6)

\psline[linewidth=2.5pt,linecolor=gray,linestyle=dotted,dotsep=1pt](5,3)(5,4)

\cnode[linewidth=0.5pt,fillstyle=solid,fillcolor=white](5,4){3pt}{v}

\cnode*[linewidth=0.5pt,fillstyle=solid,fillcolor=lightgray](5,2){2pt}{v}
\cnode*[linewidth=0.5pt,fillstyle=solid,fillcolor=lightgray](5,1){2pt}{v}
\cnode*[linewidth=0.5pt,fillstyle=solid,fillcolor=lightgray](5,0){2pt}{v}
\cnode*[linewidth=0.5pt,fillstyle=solid,fillcolor=lightgray](4,6){2pt}{v}
\cnode*[linewidth=0.5pt,fillstyle=solid,fillcolor=lightgray](4,5){2pt}{v}
\cnode*[linewidth=0.5pt,fillstyle=solid,fillcolor=lightgray](4,4){2pt}{v}
\cnode*[linewidth=0.5pt,fillstyle=solid,fillcolor=lightgray](5,3){2pt}{v}
\cnode*[linewidth=0.5pt,fillstyle=solid,fillcolor=lightgray](5,4){2pt}{v}

\rput(4.65,0.2){\small{$c_{36}$}}
\rput(4.65,1.2){\small{$c_{37}$}}
\rput(4.65,2.2){\small{$c_{38}$}}
\rput(4.65,3.2){\small{$c_{39}$}}

\rput(3.65,5.8){\small{$c_{35}$}}
\rput(3.65,4.8){\small{$c_{32}$}}
\rput(3.65,3.8){\small{$c_{31}$}}

\rput(4.65,4.2){\small{$c_{40}$}}

\rput(5.25,4.2){\small{$1$}}
\rput(4.25,3.8){\small{$?$}}
\rput(5.25,3.2){\small{$1$}}

}
\endpspicture

\vspace{0.6cm}
\caption{(a) The RSFA for the input set $P$ (white dots) with the root set $R$ (gray dots) here $h=7$.\\ (b) Since $c_{27}, c_{32} \in P \cup R$, the depicted 0/1-assignment $f_{32}$ is non $32$-eligible for points $c_{26},c_{27}, \dots, c_{32}$. (c,\,d) Eligible 0/1-assignments for points $c_{26},c_{27}, \dots, c_{32}$. Notice that the assignment $f_{32}$ depicted in (c), although  $32$-eligible, will be eventually ``useless'' for the input sets $P$ and $R$ since any RSFA for $P$ with the root set $R$ has to use point $c_{29}$ or $c_{30}$, in order to connect the point $p$ to the root $r=(0,0)$. (e,\,f) Since $c_{40} \in P$, we must have $f_{40}(c_{40}) =1 $ for any $40$-eligible function $f_{40}$. Next, in order to construct the RSFA $F_{f_{40}}$ recursively, we must have (e) $f_{39}(c_{31})=1$ or (f) $f_{39}(c_{39})=1$ for the relevant $39$-eligible function $f_{39}$. }\label{fig:FPT-gates}
\end{center}
\end{figure}

All this can be done by dynamic programming method, and the rationale for such approach is the following: given a $k$-eligible function $f_k$, an optimal RSFA $F_{f_{k}}$ consists of an optimal RSFA $F_{f_{k-1}}$ for the relevant $(k-1)$-eligible function $f_{k-1}$, being ``consistent'' with the status of the intersection point $c_{k}$, which can be '$0/R$' or '$1 / P$', in particular, the status $1/P$ of $c_{k}$ requires $f_{k-1}$ such that $c_{k-1} \in f_{k-1}^{-1}(1)$ if $k \le v$, and  $\{c_{\max(1,k-h)},c_{k-1}\} \cap f_{k-1}^{-1}(1) \neq \emptyset$ otherwise. 

To be more specific, the dynamic programming recursion (Bellman equation) is defined as follows. The base of the recursion is $k=1$, namely, for $k=1$, the only $1$-eligible function $f_1 \colon \{c_1\} \rightarrow \{0,1\}$ is the one with $f_1(c_1)=1$ (since $c_1=(0,0) \in R$), and so $F_{f_1}$ constitutes of the singleton $c_1$.  Next, for $2 \le k \le h$, if $c_k \notin R$ and $f_k$ is a $k$-eligible function with $f_k(c_k)=1$, then $F_{f_k}$ is the RSFA constructed from the RSFA $F_{f_{k-1}}$ (already tabulated) by adding the edge $c_jc_k$, where $j = \max \{i : f_{k-1}(c_i)=1\}$ and $f_{k-1}$ is the restriction of $f_{k}$ to $\{c_1,\ldots, c_{k-1}\}$, that is,  $f_k(x)=f_{k-1}(x)$ for each $x \in \{c_1,\ldots, c_{k-1}\}$. (Notice that $f_{k-1}$ is a $(k-1)$-eligible function by $k$-eligibility of $f_k$.) Otherwise, if either $c_k \in R$ or $f_k(c_k)=0$, then we just set $F_{f_k}:=F_{f_{k-1}}$, with $f_{k-1}$ defined as for the case $c_k \notin R$ and $f_k(c_k)=1$.   

Assume now that $k >h$ and let $f_k$ be a $k$-eligible function. If $c_k \notin R$ and $f_k(c_k)=1$, then the RSFA $F_{f_k}$ is constructed from the RSFA $F_{f_{k-1}}$ (already tabulated) by adding the edge $c_jc_k$, where $f_{k-1}$ is a $(k-1)$-eligible function such that $f_{k-1}(x)=f_k(x)$ for each $x \in \{c_{k-h+1},c_{k-h+2},\ldots, c_{k-1}\}$, and:
\begin{itemize}
\item[a)] $j=k-h$ if $f_{k-1}(c_{k-h})=1$, or
\item[b)] $j=k-1$ if $k \neq 1 \bmod h$ and $f_{k-1}(c_{k-1})=1$. 
\end{itemize}
Clearly, among cases (a-b) we choose the one that minimizes the sum $$l(F_{f_{k-1}})+{\rm dist}(c_j,c_k),$$ taken over all such $(k-1)$-eligible functions $f_{k-1}$. Otherwise, if either $c_k \in R$ or $f_k(c_k)=0$, we just set $F_{f_k}:=F_{f_{k-1}}$, where $f_{k-1}$ is a ($k-1$)-eligible function such that $f_{k-1}(x)=f_k(x)$ for each $x \in \{c_{k-h+1},c_{k-h+2},\ldots, c_{k-1}\}$.

The correctness of the above approach, in particular, the fact that at each single step, the constructed object is an optimal RSFA, directly follows from the construction (by induction on $k$; we omit details).
And, since we tabulate the values/RSFAs $F_{f_k}$ for each $k \in \{1,\ldots, vh\}$ and each $k$-eligible function $f_k$, which are at at most $O(2^h)$ of them for each $k$, and the RSFA $F_{f_k}$ can be computed in $O(1)$ time from all the already tabulated data, the overall time complexity of our approach is $O(vh2^h)$, and so $O((n+m)h2^h)$ since $v \le n+m$. Consequently, the RSFA problem is fixed-parameter
tractable with respect to the parameter $h$. 

\begin{theorem}
Given a set $P$ of $n$ points and a set $R$ of $m$
roots, with $(0,0) \in R$ lying on $h$ horizontal lines in the first quadrant of the plane, the RSFA for $P$ with the root set $R$ can computed in $O((n+m)h2^h)$ time and
$O(2^h)$ space.
\end{theorem}

\section{Future work} 

In this section, we discuss briefly two results that are sought to make our RFSA-survey complete. We expect that all approaches (those discussed above as well as these discussed below), following for example~\cite{R-2003}, can be extended for the case when points and roots can be located in any of four quadrants, that is, to the variant when the following constraint is imposed: for any for each point $p \in P$,  there exists a path $\pi$ to some root $r \in R$ in the same quadrant as $p$ such that the length of $\pi$ is equal to $|{\rm x}(p)-{\rm x}(r)|+|{\rm y}(p)-{\rm y}(r)|$.

\paragraph{Faster exact algorithm.} Recently, inspired by the work of Klein and Marx~\cite{KM14} on a subexponential-time algorithm for the Subset Traveling Salesman problem on planar graphs, Fomin et al.~\cite{FKLPS-2020} proposed the first subexponential algorithm for the RSA problem, with the running time of $O(2^{\sqrt{n} \log n})$; recall that the relevant algorithm in Section~\ref{sec:Exact} runs in $O(3^n)$ time.
We believe that their result can be extended to the RSFA problem as well. 

\paragraph{Efficient $2$-approximation algorithm.} 
From a practical point of view, the running time of our PTAS for the RSFA problem is unsatisfactory, which motivates us to design a fast constant factor approximation algorithm. For the RSA problem, there are two greedy 2-approximated algorithms~\cite{R-2003,RSHS-1992}. It is natural then to try to apply similar approaches for the RSFA problem.

The idea of the 2-approximation algorithm of Rao et al.~\cite{RSHS-1992} for the RSA problem is as follows. Given a set of points $P$ and the root $r$, the arborescence is generated by iteratively replacing the pair of points $p$ and $q$ in $P$ by $\pq$ until a single point (which has to be $r$) remains. The points $p$ and $q$ are chosen to maximize $||\pq||$ over all points in the current set $P$, and the resulting arborescence consists of all relevant edges from $\pq$ to $p$ and $q$ (a degenerate case is possible if $\pq$ is either $p$ or $q$). Now, for the RSFA problem, when considering a set of roots $R$, a natural idea is to apply above heuristic with the only difference that we must handle also elements in $R$ as possible roots for some trees in the current partial solution. Observe that simple replacing points $p$ and $q$ whenever it is possible does not have a reasonable relative performance guarantee as illustrated in Fig.~\ref{fig:simple_does_not_work}. Clearly, one can abstain from the relevant replacement, since the distance from $p$ to the nearest root $r_1$ is smaller then the sum of distances from $p$ to $\pq$ and from $\pq$ to $r$. This introduces, however, another principle in the (possible) correctness proof, and so another attempt is trying to rely directly on the $2$-approximation method in~\cite{R-2003}, using that algorithm as a subroutine. We believe that simple modification of this method 
%by adding the set of roots to the event heap
will result in an effective $2$-approximation algorithm for the RSFA problem as well. 

\begin{figure}
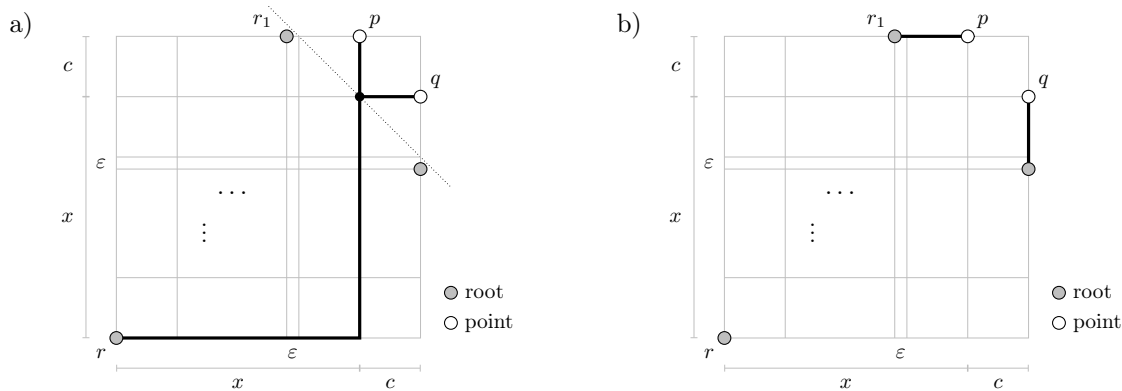

\begin{center}
\pspicture(0,-1)(8,4)
\rput(-1.25,4.15){\small{a)}}
\scalebox{0.8}{

\psline[linewidth=0.5pt,linestyle=dotted,dotsep=1pt,linecolor=black](2.5,5.5)(5.5,2.5)
\psline[linewidth=0.25pt,linecolor=lightgray](0,1)(5,1)
%\psline[linewidth=0.25pt,linecolor=lightgray](0,2.5)(5,2.5)
\psline[linewidth=0.25pt,linecolor=lightgray](0,4)(5,4)
\psline[linewidth=0.25pt,linecolor=lightgray](0,3)(5,3)
\psline[linewidth=0.25pt,linecolor=lightgray](0,2.8)(5,2.8)

\psline[linewidth=0.25pt,linecolor=lightgray](1,0)(1,5)
%\psline[linewidth=0.25pt,linecolor=lightgray](2.5,0)(2.5,5)
\psline[linewidth=0.25pt,linecolor=lightgray](4,0)(4,5)
\psline[linewidth=0.25pt,linecolor=lightgray](3,0)(3,5)
\psline[linewidth=0.25pt,linecolor=lightgray](2.8,0)(2.8,5)

\psline[linewidth=0.25pt,linecolor=lightgray](0,0)(5,0)(5,5)(0,5)(0,0)

\psline[linewidth=1.5pt,linecolor=black](4,5)(4,0)(0,0)
\psline[linewidth=1.5pt,linecolor=black](4,4)(5,4)

\cnode*[linewidth=0.5pt,fillstyle=solid,linecolor=black](4,4){2.25pt}{v}

\cnode[linewidth=0.5pt,fillstyle=solid,fillcolor=lightgray](0,0){3pt}{v}
\cnode[linewidth=0.5pt,fillstyle=solid,fillcolor=white](4,5){3pt}{v}
\cnode[linewidth=0.5pt,fillstyle=solid,fillcolor=white](5,4){3pt}{v}
\cnode[linewidth=0.5pt,fillstyle=solid,fillcolor=lightgray](2.8,5){3pt}{v}
\cnode[linewidth=0.5pt,fillstyle=solid,fillcolor=lightgray](5,2.8){3pt}{v}

\cnode[linewidth=0.5pt,fillstyle=solid,fillcolor=lightgray](5.5,0.75){3pt}{v}
\cnode[linewidth=0.5pt,fillstyle=solid,fillcolor=white](5.5,0.25){3pt}{v}
\rput(6.13,0.25){\small{point}}
\rput(6.05,0.78){\small{root}}

%\rput(2.5,.75){\small{the optimal}}
%\rput(3.75,2.25){\rotateleft{\small{solution for MRSFA}}}

%\rput(2.5,2){\large$\cdots$}
%\rput(3,2.65){\large$\vdots$}
\rput(1.9,2.4){\large$\cdots$}
\rput(1.45,1.85){\large$\vdots$}

\rput(-.25,-.25){\small$r$}
\rput(2.4,5.25){\small$r_1$}
\rput(4.25,5.25){\small$p$}
\rput(5.25,4.25){\small$q$}

\psline[linewidth=0.25pt,linecolor=lightgray]{|-|}(0,-0.5)(4,-.5)
\rput(2,-.75){\small$x$}
\psline[linewidth=0.25pt,linecolor=lightgray]{|-|}(-0.5,0)(-0.5,4)
\rput(-0.8,2){\small$x$}

\rput(2.9,-0.25){\small$\eps$}
\rput(-0.25,2.9){\small$\eps$}

\psline[linewidth=0.25pt,linecolor=lightgray]{|-|}(5,-0.5)(4,-.5)
\rput(4.5,-.75){\small$c$}
\psline[linewidth=0.25pt,linecolor=lightgray]{|-|}(-0.5,5)(-0.5,4)
\rput(-0.8,4.5){\small$c$}

}
\endpspicture
\pspicture(0,-1)(4,4)
\rput(-1.25,4.15){\small{b)}}
\scalebox{0.8}{

\psline[linewidth=0.25pt,linecolor=lightgray](0,1)(5,1)
%\psline[linewidth=0.25pt,linecolor=lightgray](0,2.5)(5,2.5)
\psline[linewidth=0.25pt,linecolor=lightgray](0,4)(5,4)
\psline[linewidth=0.25pt,linecolor=lightgray](0,3)(5,3)
\psline[linewidth=0.25pt,linecolor=lightgray](0,2.8)(5,2.8)

\psline[linewidth=0.25pt,linecolor=lightgray](1,0)(1,5)
%\psline[linewidth=0.25pt,linecolor=lightgray](2.5,0)(2.5,5)
\psline[linewidth=0.25pt,linecolor=lightgray](4,0)(4,5)
\psline[linewidth=0.25pt,linecolor=lightgray](3,0)(3,5)
\psline[linewidth=0.25pt,linecolor=lightgray](2.8,0)(2.8,5)

\psline[linewidth=0.25pt,linecolor=lightgray](0,0)(5,0)(5,5)(0,5)(0,0)

\psline[linewidth=1.5pt,linecolor=black](4,5)(2.8,5)
\psline[linewidth=1.5pt,linecolor=black](5,4)(5,2.8)

\cnode[linewidth=0.5pt,fillstyle=solid,fillcolor=lightgray](0,0){3pt}{v}
\cnode[linewidth=0.5pt,fillstyle=solid,fillcolor=white](4,5){3pt}{v}
\cnode[linewidth=0.5pt,fillstyle=solid,fillcolor=white](5,4){3pt}{v}
\cnode[linewidth=0.5pt,fillstyle=solid,fillcolor=lightgray](2.8,5){3pt}{v}
\cnode[linewidth=0.5pt,fillstyle=solid,fillcolor=lightgray](5,2.8){3pt}{v}

\cnode[linewidth=0.5pt,fillstyle=solid,fillcolor=lightgray](5.5,0.75){3pt}{v}
\cnode[linewidth=0.5pt,fillstyle=solid,fillcolor=white](5.5,0.25){3pt}{v}
\rput(6.13,0.25){\small{point}}
\rput(6.05,0.78){\small{root}}

%\rput(2.5,.75){\small{the optimal}}
%\rput(3.75,2.25){\rotateleft{\small{solution for MRSFA}}}

%\rput(2.5,2){\large$\cdots$}
%\rput(3,2.65){\large$\vdots$}
\rput(1.9,2.4){\large$\cdots$}
\rput(1.45,1.85){\large$\vdots$}

\psline[linewidth=0.25pt,linecolor=lightgray]{|-|}(0,-0.5)(4,-.5)
\rput(2,-.75){\small$x$}
\psline[linewidth=0.25pt,linecolor=lightgray]{|-|}(-0.5,0)(-0.5,4)
\rput(-0.8,2){\small$x$}

\rput(2.9,-0.25){\small$\eps$}
\rput(-0.25,2.9){\small$\eps$}

\psline[linewidth=0.25pt,linecolor=lightgray]{|-|}(5,-0.5)(4,-.5)
\rput(4.5,-.75){\small$c$}
\psline[linewidth=0.25pt,linecolor=lightgray]{|-|}(-0.5,5)(-0.5,4)
\rput(-0.8,4.5){\small$c$}

\rput(-.25,-.25){\small$r$}
\rput(2.5,5.25){\small$r_1$}
\rput(4.25,5.25){\small$p$}
\rput(5.25,4.25){\small$q$}

}
\endpspicture

\caption{(a) Replacing points $p$ and $q$ with $u=\pq$ results in a rectilinear Steiner arborescence of the cost $2(c+x)$, while (b) the cost of the optimal solution is $2(c+\eps)$.}\label{fig:simple_does_not_work}
\end{center}
\end{figure}

\section{Open problem --- Manhattan networks}
Given a set of points $P$ in $\mathbb{R}^2$, a network $(G,P)$ is said to be a {\sl Manhattan network} for $P$ if for all $p,q \in P$, there exists a Manhattan path between $p$ and $q$ with all its segments in $G$; the length of $(G,P)$, denoted by $L(G)$, is the total length of all segments in $G$. For the given set $S$,  the {\sl Minimum Manhattan Network} (MNN) problem is to find a Manhattan network $(G,S)$ of minimum length $L(G)$.
The MNN problem was first introduced by Gudmundson et al.~\cite{GLN-2001}, who proposed an $O(n^3)$-time $4$-approximation algorithm and an $O(n \log n)$-time $8$-approximation algorithm.  After~\cite{GLN-2001}, much research has been devoted to finding better approximation algorithms --- see for example~\cite{BWW-2004, BWWT-2006, ChNV-2008, GSZ-2011, KIT-2002, SU-2005} --- and only in $2009$, Chin et al.~\cite{CGS-2009} proved the MMN problem to be NP-hard. In addition, 
a fixed-parameter algorithm running in $O^{\ast}(2^{14h})$ time (neglecting a factor that is polynomial in $n$), where the parameter $h$ is the minimum number of axis-parallel straight-lines (either all horizontal or all vertical) that contain all the points in $P$, was proposed by Knauer and in~\cite{KS-2011}, the MNN problem in higher dimensions has been considered in~\cite{D+2015,MSU-2009}, while \cite{G+2007, JMR-2022} consider a variant of the problem where the goal is to minimize the number of vertices (Steiner points) and edges. 
    
Keeping in mind the motivation of RSFA, it is natural to define a relevant variant of MNN. Namely, given a minimum Manhattan network $(G_1,P_1)$ and the input set $P_2$ of points within the bounding box of $P_1$, the {\sl Minimum Manhattan Network with the Backbone $G_1$} (MMNB) problem is to compute the minimum Manhattan network $(G_2, P_1 \cup P_2)$ such that ${\rm Un}(G_1) \subseteq {\rm Un}(G_2)$, where  ${\rm Un}(G)$ denotes the union of all line segments (edges) in a Manhattan network $G$. In particular, one can be interested in designing a fast $2$-approximation algorithm for the MMNB problem.

\end{document}